\documentclass{CEAS_EuroGNC_2026}

\usepackage[utf8]{inputenc}

\usepackage{graphicx}
\usepackage{amsmath}
\usepackage[version=4]{mhchem}
\usepackage{siunitx}
\usepackage{array}
\usepackage{longtable}
\usepackage{framed}
\usepackage[most]{tcolorbox}
\newtcolorbox{tcbdoublebox}[1][]{%
  enhanced jigsaw,
  sharp corners,
  colback=white,
  borderline={1pt}{-2pt}{black},
  fontupper={\setlength{\parindent}{20pt}},
  #1
}

 \usepackage{subfigure}
\usepackage{float} 
\usepackage{multicol} 
\usepackage{mathtools}
\usepackage{hhline} 

\newtheorem{theorem}{Theorem}
\newtheorem{lemma}{Lemma}

\newtheorem{remark}{Remark}

\newtheorem{assumption}{Assumption}

\newtheorem{opt}{Optimization}

\newcommand{\sat}{\mathrm{sat}}
\newcommand{\sign}{\mathrm{sign}}

\newenvironment{proof}{\noindent {\em Proof.}}{\hfill \hspace*{1pt} \hfill $\square$}

\title{A Switching Strategy for Event-Trigger Control of Spacecraft Rendezvous}
\papernumber{671777}

\ifpdf
\hypersetup{pdfinfo={
Title={A Hybrid Switching Strategy for Event-Trigger Control of Spacecraft Rendezvous},
Author={Tommaso Del Carro, Alexandre Seuret, Rafael Vazquez},
Keywords={Spacecraft Rendezvous, Event-trigger control, Linear Matrix Inequalities, Lyapunov
control theory, Guidance, navigation and control of space vehicles, Autonomous systems} 
}}
\fi

\begin{document}
\thispagestyle{fancy}

\maketitle
\pagestyle{empty} 

\begin{authorList}{4cm} 
\addAuthor{Tommaso Del Carro}{PhD student, Escuela Tecnica Superior de Ingeniería, Universidad de Sevilla, Departamento de Ingeniería de Sistemas y Automática, Seville, Spain. \emailAddress{tdelcarro@us.es}}
\addAuthor{Gerson Portilla}{ 
Researcher, Escuela Tecnica Superior de Ingeniería, Universidad de Sevilla, Departamento de Ingeniería de Sistemas y Automática, Seville, Spain. \emailAddress{gportilla@us.es}}
\addAuthor{Alexandre Seuret}{ 
Researcher, Escuela Tecnica Superior de Ingeniería, Universidad de Sevilla, Departamento de Ingeniería de Sistemas y Automática, Seville, Spain. \emailAddress{aseuret@us.es}}
\addAuthor{Rafael Vazquez}{Full Professor, Escuela Tecnica Superior de Ingeniería, Universidad de Sevilla, Departamento de Ingeniería Aeroespacial, Seville, Spain. \emailAddress{rvazquez1@us.es}}
\end{authorList}
\justifying

\begin{abstract}
This paper presents the design of a state-feedback control law for spacecraft rendezvous, formulated using the Hill–Clohessy–Wiltshire (HCW) equations. The proposed method introduces an impulsive control strategy to regulate thruster operations. Specifically, a state-dependent switching framework is developed to determine both the control input magnitudes and the precise state conditions that trigger thruster activation. The nonlinear control law is derived using principles from automatic control theory, particularly Lyapunov stability analysis and the Linear Matrix Inequality (LMI) framework. The resulting closed-loop system is proven to be stable, while simultaneously minimizing the total number of actuation events. The effectiveness of the proposed method is demonstrated through a numerical case study, which includes a comparative analysis with a standard Model Predictive Control (MPC) scheme, highlighting the advantages and trade-offs of the developed control structure.

\end{abstract}


\section*{Nomenclature}

{\renewcommand\arraystretch{1.0}
\noindent\begin{longtable*}{@{}l @{\quad=\quad} l@{}}
$r, R$  & position\\
$v, u$ & velocity\\
$\mu$ & Earth gravitational parameter\\
$n$ & mean motion\\
\end{longtable*}}
\newpage 
\section{Introduction}

Nowadays, spacecraft rendezvous has become a critical application in a wide set of meaningful space missions involving proximity operations such as on-orbit assembly, exploration of small bodies, and space debris removal. In the context of lunar missions, an increasingly urgent area of interest, rendezvous has long played a vital role. It was crucial to the development of the Apollo program \citep{NEUFELD2008540} and now serves as a cornerstone of NASA Artemis initiative \citep{doi:10.2514/6.2020-1921}. A central element of this effort is the Lunar Orbital Platform-Gateway \citep{doi:10.2514/6.2018-5337}, a cislunar station designed as a communication and scientific hub.\\
\indent Extensive literature exists on spacecraft rendezvous. The foundation of relative motion analysis lies in the Hill–Clohessy–Wiltshire (HCW) equations \citep{ee197d15-81b3-3df1-9469-8da0b5107e60}, \citep{doi:10.2514/8.8704}, which model close-proximity motion in circular orbits and admit analytical solutions.\\
\indent Efficient control of relative motion remains challenging, requiring precise, fuel-efficient maneuvers under tight operational constraints. Developments in automatic control theory have facilitated the systematic optimization of trajectories while guaranteeing closed-loop stability and desired performance. Increasing degrees of onboard autonomy have become indispensable, as reliance on ground-based command is rendered impractical by inherent communication delays \citep{Fehse_2003} and the escalating density of active satellites in Low Earth Orbits (LEO). Autonomous control architectures enable prompt and adaptive responses to dynamic mission conditions, thereby enhancing the safety, robustness, and overall efficiency of rendezvous and proximity operations.\\
\indent The wide range of possible mission constraints has motivated extensive use of control techniques such as Model Predictive Control (MPC) \citep{https://doi.org/10.1002/rnc.2827}, \citep{HARTLEY2012695}, \citep{7012053}. MPC operates by solving multiple optimization problems over a finite prediction horizon, in order to minimize a cost function, typically related to fuel consumption (see \citep{MAYNE2000789} for foundational concepts). Several enhancements have addressed uncertainties and complex scenarios, including tumbling-target rendezvous with collision avoidance \citep{LI2017700}, \citep{DONG2022176} and chance-constrained formulations for LEO \citep{GAVILAN2012111} and Earth-Moon L2 Near-Rectilinear Halo Orbit operations \citep{SANCHEZ2020105827}. Continuous-time constraint enforcement via polynomial non-negativity was explored by \citep{6314659}, \citep{doi:10.2514/1.G000283}, and \citep{GILZ20177229}, and later extended to a six-degree-of-freedom rendezvous by exploiting the attitude flatness property \citep{SANCHEZ2020391}. Despite its versatility, MPC faces challenges related to stability guarantees, the proper handling of minimum impulse bit constraints, and the high computational cost of real-time optimization. These are limitations that are particularly critical for resource-constrained platforms such as CubeSats.\\
\indent To address these limitations, we adopt an alternative strategy based on the \textit{event-triggered control paradigm} (ETC) \citep{Astrom2008}, \citep{6425820}, which applies control inputs only when state-dependent conditions are met. This approach minimizes unnecessary actuation and reduces computational and energy demands. Our method builds on the HCW model for circular orbits, focusing on close-range rendezvous phases, exploiting \textit{Lyapunov-based stability analysis and design} for a class of \textit{Hybrid Dynamical Systems} \citep{693dbdb4-c870-3adb-a23e-a58722bf715d}, naturally suited to systems combining continuous evolution with discrete events, such as those caused by instantaneous control inputs or abrupt state transitions.\\
\indent Recent research in the hybrid control community has addressed spacecraft applications: \citep{8331888} treated elliptical rendezvous via Floquet–Lyapunov transformations; \citep{seuret2024hybriddynamicalapproachimpulsive} employed the HCW model and demonstrated stability separately for the in-plane and out-of-plane dynamics using Lyapunov methods; \citep{Sanchez_2021} proposed event-based predictive control for rendezvous hovering phases, where triggering conditions are defined by the admissible set reachability; \citep{YANG2015454} implemented an event-based control scheme based on a switching logic, distinguishing between open-loop and closed-loop dynamics governed by periodically applied impulses of finite duration; \citep{7798769} presented robust hybrid supervisory control to coordinate phase-specific controllers designed for different stages of the rendezvous and docking mission and \citep{CRANE201894} validated these results in a real flight software system tested inside a software-in-the-loop environment; \citep{XIE2025918} adopted event-trigger impulsive control for the station-keeping of a near-asteroid spacecraft, illustrating its capability to reduce the energy cost and extend the spacecraft’s lifetime; \citep{ZHANG20182620} and \citep{WU2018927} recently exploited event-based strategies to deal with attitude control applications.\\
\indent In this work, the triggering mechanism is governed by a tuning parameter, which selects both the control gains and the regions of the state space where thrusters need to be activated. This approach enables the control inputs of the chaser to be applied only when necessary and not at all sampling times, resulting in asymptotic convergence to the target position with reduced fuel consumption. Under chemical propulsion, maneuvers are often modeled as instantaneous velocity variations. However, physical limitations impose bounds on the maximum thrust delivered by each actuator. Therefore, the control law, derived via Lyapunov theory, explicitly accounts for input saturation \citep{tarbouriech2011stability, SEURET2024101045}, as well as process noise \citep{10468557}. A key advantage is that the underlying feasibility problem, expressed in the form of Linear Matrix Inequalities (LMIs), is solved offline only once, thus minimizing onboard computation. This hybrid, event-triggered framework provides a robust, low-complexity alternative for autonomous rendezvous and opens perspectives for future developments, such as adaptation to elliptical orbits. The proposed strategy can also be viewed as applicable to hovering or formation-keeping scenarios (see \citep{Sanchez_2021}), where the objective is to remain near a point at a given distance from the target spacecraft. In such cases, safety requirements are considerably relaxed, making our formulation well-suited for this class of problems.\\
\indent This paper is structured as follows. Section \ref{sec:problem_statement} introduces the Hill-Clohessy-Wiltshire relative dynamics with impulsive inputs and states the control objectives. Section \ref{sec:ETC} presents the proposed event-triggered control strategy based on a switched system framework. Section 
\ref{sec:section_stabilization} provides the main result of this paper, in the form of a stabilization theorem for the closed-loop dynamics under the selected control law. Section \ref{sec:simulations} presents numerical applications illustrating the effectiveness of the control law in various close-range rendezvous scenarios and provides comparisons with more standard techniques, such as MPC. Finally, Section \ref{sec:conclusions} concludes the paper, providing remarks and discussing potential directions for future research.\\
\indent\textbf{Notations:} Throughout the paper, $\mathbb R^{n\times m}$ denotes the set of all real $n \times  m$ matrices, $\mathbb S^{n}$ ($\mathbb D^n$) the set of symmetric (diagonal) matrices in $\mathbb R^{n\times n}$. Matrices $I$ and ${0}$  denote the identity and null matrices of the appropriate dimensions, respectively. For any matrices $A=A^{\!\top},\ B,\ C=C^{\!\top}$ of appropriate dimensions, matrix $\left[\begin{smallmatrix}A&B\\\ast & C  \end{smallmatrix} \right]$ denotes the symmetric matrix $\left[\begin{smallmatrix}A&B\\ B^{\!\top}& C  \end{smallmatrix} \right]$. 
For any matrix $M\in\mathbb R^{n\times n}$, the notation  $M\succ0$ means that $M\in\mathbb S^{n}_+$. 
The notation $(Z)_\ell$ for any matrix $Z\in\mathbb R$ indicates the $\ell$-th row of $Z$. Finally, for any matrix $A$, the Hermitian component of such a matrix is defined as $\text{He}(A)=A+A^{\top}$.

\section{Problem statement}\label{sec:problem_statement}

\subsection{HCW model with saturated impulsive actuation}
A variety of mathematical models can be employed for spacecraft rendezvous. Specifically, when the target spacecraft is situated in a circular Keplerian orbit and the approaching vehicle, referred to as the chaser, is in close proximity, the linear HCW equations, introduced in \citep{ee197d15-81b3-3df1-9469-8da0b5107e60} and \citep{doi:10.2514/8.8704}, offer an accurate representation of the relative positions and velocities of the spacecraft.
To establish the spatial configuration of the spacecraft rendezvous, the relative motion between the chaser and the target is illustrated in Figure~\ref{fig0}. The $XYZ$ axes define the \textit{Geocentric Equatorial Frame}, centered at the center of mass of the Earth. The local frame of the target, also denoted as \textit{Local Vertical-Local Horizontal} (LVLH) frame, is centered in point A. Its $x$ axis coincides with the outward radial from the Earth to the satellite itself. The $z$ axis is perpendicular to the orbital plane, and the $y$ axis, which is tangential to the target orbit only when it is circular, completes the reference frame. In Figure~\ref{fig1}, vectors connecting the Earth to the target (point A) and the chaser (point B) are denoted as $R$ and $R_0$, respectively. The vector $r=R-R_0$ connecting A to B coincides with the relative position of the chaser in the LVLH frame.
\begin{figure}[htbp]
    \centering
    \begin{minipage}[b]{0.30\textwidth}
        \centering
        \includegraphics[width=\textwidth]{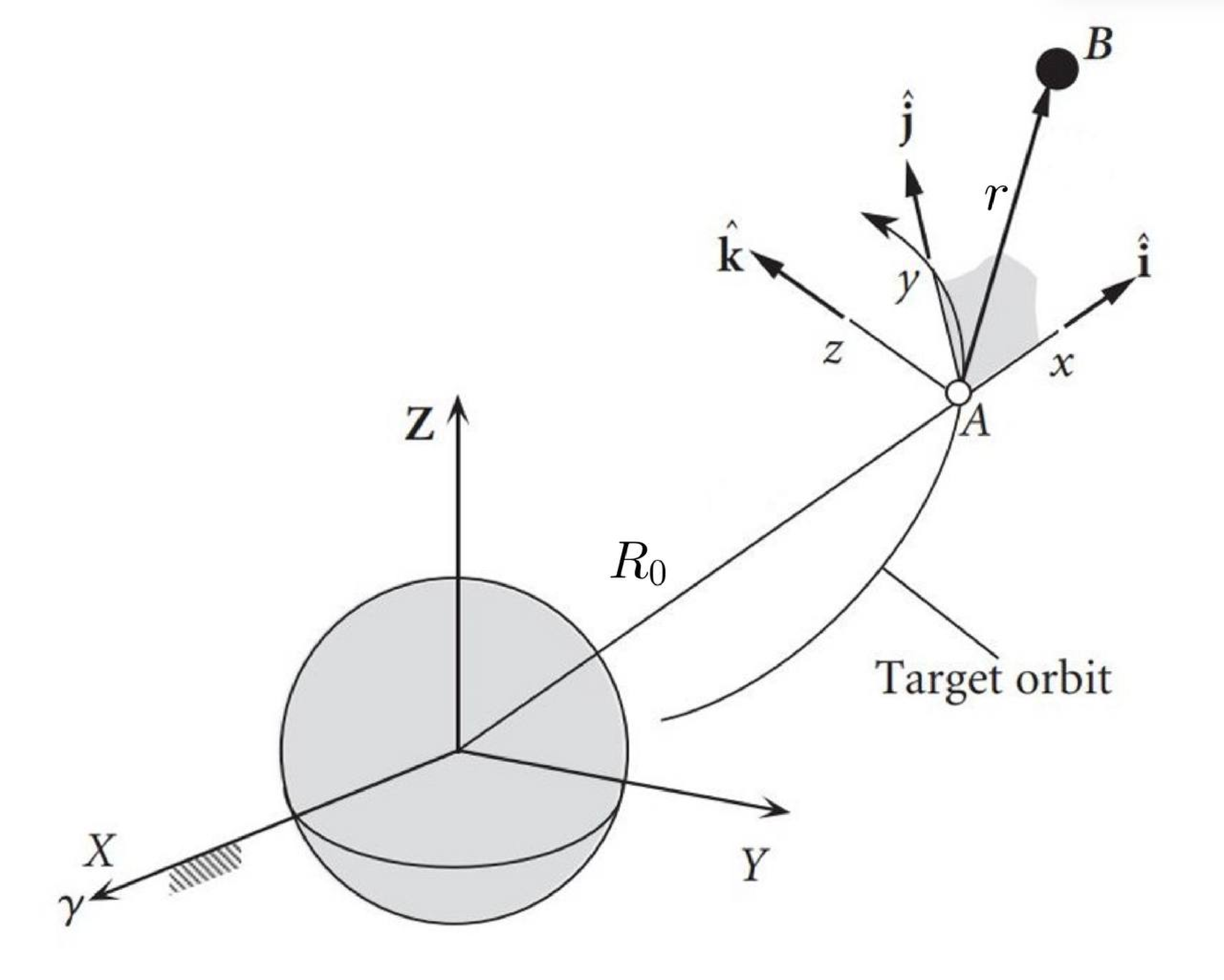}
        \caption{Inertial and LVLH reference frames. Adapted from \cite{CURTIS2014367}}
        \label{fig0}
    \end{minipage}
\hspace{0.15\textwidth}
    \begin{minipage}[b]{0.27\textwidth}
        \centering
        \includegraphics[width=\textwidth]{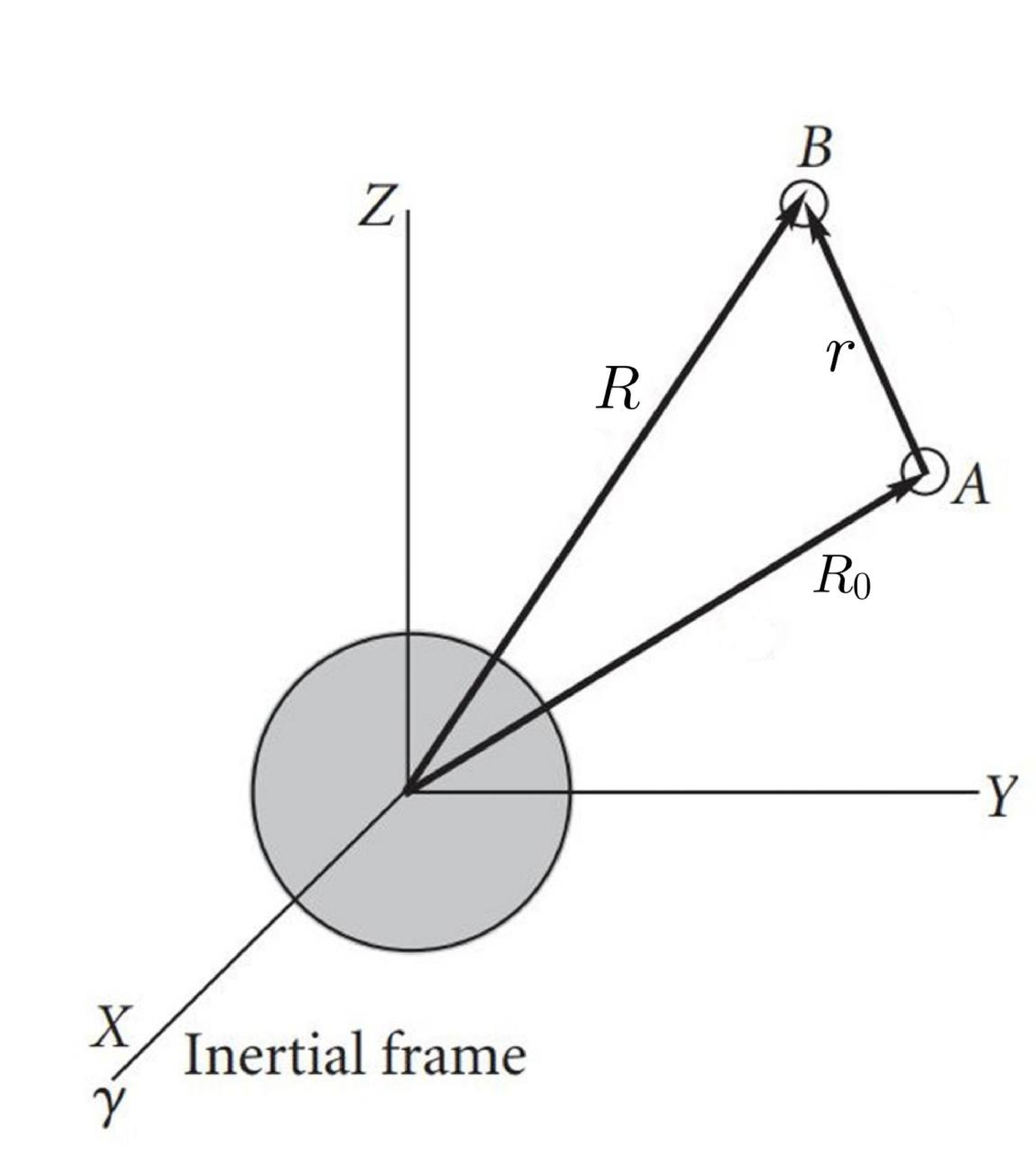}
        \caption{Chaser and target position in the LVLH frame. Adapted from \cite{CURTIS2014367}}
        \label{fig1}
    \end{minipage}
\end{figure}
%
The derivation of the HCW equations starts from the study of the two-body problem (2BP) of both spacecraft, that in the case of the chaser is expressed as
\begin{equation}\label{2bp}
    \Ddot{R} = -\mu\frac{R}{\| R\|^3},
\end{equation}
where $\mu$ is the standard gravitational parameter of the central body, for the Earth $\mu= 398600.4 \frac{km^3}{s^2}$.\\ 
\indent The HCW model is based on two main hypotheses: first, to consider a circular target orbit. Second, to assume that the distance between the chaser and the target is significantly small in comparison to their distance from the Earth, i.e., \(\frac{\|r\|}{\|R_0\|} \ll 1\). Thus, the resulting set of HCW equations is given by:
\begin{equation}
\label{eq:HCW}
\left\{
\begin{aligned}
     \ddot{r}_x\!-\!2n \dot{r}_y\!-\!3n^2 r_x&\!=\!w_x,\\
     \ddot{r}_y\!+\!2n\dot{r}_{x}&\!=\!w_y,\\
    \ddot{r}_z\!+\!n^2r_z&\!=\!w_z,
\end{aligned}
\right. 
\end{equation}
where $r = [\,r_x \; r_y \; r_z\,]^{\top}$ and $
\dot{r} = v = [\,v_x \; v_y \; v_z\,]^{\top}$ denote the relative positions and velocities between the chaser and the target in the target reference frame, respectively. Dynamics \eqref{eq:HCW} includes a disturbance vector $w = [\,w_x \; w_y \; w_z\,]^{\top}$, aiming to represent the unmodeled dynamics (due to linearization) or additive external perturbations. The assumptions on this vector will be provided later on. Here, $n$ stands for the so-called mean motion, representing the mean (instantaneous for circular orbits) angular velocity of an orbiting spacecraft around the central body. For a circular orbit of radius $\|R_0\|$ we have $n=(\mu/\|R_0\|^3)^{\frac{1}{2}}$. Thereby, for a typical Low Earth Orbit of radius $\|R_0\|$ = 500 km, we would have $n$ = 0.0011 rad/s. It is worth mentioning that the $r_z$ motion, which represents the out-of-plane dynamics, is decoupled with respect to the in-plane one, $r_x-r_y$. The drift in the $r_y$ direction, which is slow in nature, characterizes a weakly unstable in-plane dynamics.\\
\indent Adopting the choice of chemical thrusters for the chaser propulsion, it is common practice to approximate the inputs as velocity impulses. According to this model, the thrust is able to alter instantaneously the velocity states of the chaser, while the positions are assumed to be unaffected. More precisely, the velocities follow the difference equations
\begin{equation}
\label{eq:HCW_imp_control}
\left\{
\begin{aligned}
     {v}_x^+=v_x+\sat(u_x),\\
     {v}_y^+=v_y+\sat(u_y),\\
     {v}_z^+=v_z+\sat(u_z),\\
\end{aligned}
\right. 
\end{equation}
where the $v_x^+$, $v_y^+$ and $v_z^+$ stand for the value of the velocities $v_x$, $v_y$ and $v_z$ immediately after an impulse. In this equation, $u_x$, $u_y$ and $u_z$ are the control inputs to be designed. To satisfy practical constraints, each control input is subject to saturation, limiting the maximum thrust that each actuator can produce.  Each thruster is assumed to take values only within the bounded interval $[-\bar u, \bar u]$, where $\bar u>0$ is a known constant common to all thrusters, resulting in a saturated impulsive control input. Finally, the body axes of the vehicle are assumed to be always aligned to the LVLH frame. Therefore, the velocity saturation can be directly expressed in this frame. The saturation function is given by:
    \begin{equation}
        \sat(u):= \begin{bmatrix}
            \sat(u_x)\\
            \sat(u_y)\\
            \sat(u_z)\\
        \end{bmatrix}=\begin{bmatrix}
            \min(\bar u, |u_x|)\cdot\sign(u_x)\\
            \min(\bar u, |u_y|)\cdot\sign(u_y)\\
            \min(\bar u, |u_z|)\cdot\sign(u_z)\\
        \end{bmatrix}.
    \end{equation}

\subsection{Formulation of the rendezvous control problem}

Considering that the target occupies the center of the LVLH frame, the control objective of guiding the chaser to the target position translates into the stabilization to the origin. This goal has to be achieved while accounting for input saturation and process noise, whose presence represents an element of complexity in the stability properties of the closed-loop system. With the crucial goal of reducing control energy consumption, our objective is to include a so-called \textit{event-triggered control} \cite{6425820}. The main advantage of this approach is found in the fact that, depending on the system state, the control is able, through the incorporation of logic decision variables, to choose whether to generate or not an impulse at a certain instant. The proposed event-triggered control strategy has the following rationale:
\begin{itemize}
\item \textit{How much?}  The control law decides the amplitude of the input to be implemented at a given instant. 
\item \textit{When?} The control law decides, depending on the current state, whether or not a control actuation is required at the considered time.
\end{itemize}

This strategy offers two main advantages. First, it introduces nonlinear behavior into the closed-loop dynamics, thereby overcoming the inherent limitations of linear control. Second, it can potentially reduce the total number of actuation events compared to Model Predictive Control (MPC) \citep{GAVILAN2012111, DONG2022176}, where a control input is computed and applied at each sampling time within the considered time horizon.\\
\indent In the next section, we present the structure of the nonlinear event-triggered control law, which results in a state-dependent switching model.
 






\section{State-dependent switching model}\label{sec:ETC}

Before exploring the design of our control law, let us first define the state vector of the system $x=[r_x,r_y,r_z,v_x,v_y,v_z]^{\! \top}\!$, that collects the relative positions and velocities of the chaser with respect to the target. Then, let us introduce the discretized version of the dynamics of system \eqref{eq:HCW}-\eqref{eq:HCW_imp_control}, obtained by a simple integration over a constant sampling period $T>0$, which yields
\begin{equation}\label{def:DSsys0}
     x^+=A x+B\text{sat}(u)+B_w w,
\end{equation}
where $A = e^{A_c T}$, $B = e^{A_c T}B_c$ and $B_w=\int_0^Te^{A_c (T-\tau)}\mathrm{d}\tau B_{w_c}$. Matrices $A_c$, $B_c$ and $B_{w_c}$ are respectively the state, input and noise input matrices of the HCW continuous system:
\begin{equation}\label{hybrid_model}
    A_c=\begin{bsmallmatrix}
    0 &0 &0 &1 &0&0\\
    0 &0 &0 &0 &1&0\\
    0 &0 &0  &0 &0&1\\
    3n^2 &0 &0  &0 &2n &0\\
    0 &0 &0  &-2n &0&0\\
    0 &0 &-n^2  &0 &0&0\\
\end{bsmallmatrix},\quad B_c=\begin{bsmallmatrix}
    0 &0 &0 \\
    0 &0 &0 \\
    0 &0 &0 \\
1 &0&0\\
0 &1&0\\
0 &0&1\\
\end{bsmallmatrix},\quad B_{w_c}=\begin{bsmallmatrix}
    0 &0 &0 &0 &0 &0 \\
    0 &0 &0 &0 &0 &0 \\
    0 &0 &0 &0 &0 &0 \\
0 &0&0&1&0&0\\
0 &0&0&0&1&0\\
0 &0&0&0&0&1\\
\end{bsmallmatrix}.
\end{equation}

\indent The continuous disturbance term $w$ in \eqref{eq:HCW} is assumed to vary slowly over time compared to the dynamics of the system. Hence, we reasonably can assume it to  be constant over every sampling period $T$. 
Moreover, $w$ is modeled as a norm-bounded white noise process. This assumption is formally expressed by: 
\begin{assumption}\label{th:bounded_noise}
There exists $\lambda > 0$ such that the norm of the disturbance verifies $w^{\top}w < \lambda$, at all times. In other words, the disturbance belongs to the following set $\Omega_{\lambda}$
\begin{equation}
    \Omega_{\lambda}  :=\{w\in\mathbb R^{n_x},\quad \mbox{ s.t.}\quad  w^\top w\leq \lambda \}\nonumber.\\
\end{equation}
\end{assumption}
\indent The switching nature of our chosen logic lies in the specific structure of the control law, which is defined as follows:
\begin{equation}\label{def:u}
    u(x)={\sigma}(x) Kx,
\end{equation}
where \begin{itemize}
\item $Kx$ infers the amplitude of the control input to be implemented , answering to question \textit{``how much?"}. Therefore, $K\in\mathbb R^{n_u\times n_{x}}$ represents a classical state-feedback control gain.

    \item $\sigma(x)$ is a logical variable that can only assume values in $\{0,1\}$. Its rationale is to specify whether the control action is performed ($\sigma(x)=1$) or not ($\sigma(x)=0$). Following previous works on event-triggered control \cite{Astrom2008} and on switched systems \cite{6425820}, we introduce the switching logic $\sigma(x) \in \{0,1\}$ as follows
\begin{equation}
\label{def:trig_indexed}
\sigma(x):=
\begin{cases}
0, & \text{if }\  x^\top\! M x \geq 0, \\
1, & \text{if }\  x^\top\! M x \leq 0,
\end{cases}
\end{equation}
where $M=M^\top$ is a symmetric matrix that allows to answer to the question \textit{``when?"}.

\end{itemize}  

\begin{remark}
In the case where $x^\top M x=0$, any solution $\sigma(x)=1$ or $\sigma(x)=0$ will be acceptable.    
\end{remark}

\begin{remark}
Consider the case where $M$ is a positive (respectively, negative) semi-definite matrix. 
In this situation, the set of points $x$ for which $\sigma(x) = 1$ (respectively, $\sigma(x) = 0$) 
reduces to the origin. Consequently, the control law would demand an impulse (respectively, no impulse) 
for any nonzero state. Therefore, to exploit the event-triggered nature of the proposed scheme, 
the matrix $M$ must possess eigenvalues of differing signs.
\end{remark}

Altogether, this control law can be interpreted as an event-triggered strategy where the control actuation combines a decision on both the amplitude and the time instants. Including this new feature, the closed-loop system is now governed by the following nonlinear and switched dynamics at the sampling times:
\begin{equation}\label{def:DSsys}
     x^+=Ax +B\sat(\sigma(x)Kx)+B_ww,
\end{equation}
where we recall that the state dependent function $\sigma(x)$ is given in \eqref{def:trig_indexed}. The main issue addressed in the next section concerns the co-design of the triggering matrix $M$ and the feedback gain $K$.

\section{Event-triggered control design}\label{sec:section_stabilization}

\subsection{Preliminary lemmas}

Before presenting the main result of this paper, useful preliminary results are introduced to facilitate subsequent development. The first one deals with a specific property of the triggering law, while the second refers to the generalized sector condition, an useful lemma to treat saturated systems. 
\begin{lemma}\label{lem:M2}
    For any $x$ in $\mathbb R^{n_x}$, the following inequality holds
    \begin{equation}\label{lem:condM_indexed}
\mathcal{L}_2(x) :=  x^\top \mathcal{M}_{\sigma(x)} x \geq 0.
\end{equation}
where
\begin{equation*}
    \mathcal{M}_{\sigma(x)}:= (1 - 2\sigma(x)) M, \quad \forall x\in\mathbb R^{n_x}.
\end{equation*}
\end{lemma}
\begin{proof}
The proof is based on the observation that
\begin{itemize}
    \item $\sigma(x) =0$ implies $x^{\! \top}  M x \geq 0$ and $1-2\sigma(x)=1\geq0$;
    \item $\sigma(x) =1$ implies $x^{\! \top}  M x \leq 0$ and $1-2\sigma(x)=-1\leq 0.$
\end{itemize} 

Therefore, we get $(1-2\sigma(x))x^\top Mx\geq0$, which concludes the proof.
\end{proof}

\indent It is well-known that the introduction of input saturation constraints requires special care in the study of the stability of a system (see \citep{tarbouriech2011stability}). Following the generalized sector condition introduced in \cite{tarbouriech2011stability}, let us define the so-called dead-zone function as $\phi(u)=\text{sat}(u)-u$. Next, we introduce a key lemma for stabilizing system~\eqref{def:DSsys0}, considering the dead-zone function.  
\begin{lemma}{\cite{tarbouriech2011stability}}\label{lem:M1}
For any matrix $G\in\mathbb R^{n_u\times n_{x}}$, define the set $\mathcal{S}(G) = \left\{ x \in \mathbb{R}^{n_{x}} : |(G)_{\ell} x| \leq \bar{u}, \ \forall \ell = 1, ..., n_u \right\}$. Then, the following relation holds for any matrix $T\in\mathbb D^{n_u}_{+}$ 
\\\begin{equation}\label{eq:ineq_sat}
\mathcal{L}_1(x) :=\phi(u)^\top T \left[ \text{sat}(u) + Gx \right] \leq 0, \quad \forall x \in \mathcal{S}(G).
\end{equation}
\end{lemma}

\subsection{Main result} 
We are now in position to state the main result of this paper:
\begin{theorem}\label{th:Stabb}
Let Assumption~\ref{th:bounded_noise} hold. For a given parameter $\mu \in (0,1)$, assume that there exist decisions variables $\mathcal D$ as follows
\begin{align*}
    \mathcal{D}&:=\{\varepsilon,W_i,{Q}, S, Y, Z_i,H\}_{i\in\{0,1\}}\in \mathbb H^{n_x}:=\mathbb{R}_{>1}\times(\mathbb{S}^{n_x})^3\times\mathbb{D}_{+}^{n_x} \times(\mathbb{R}^{n_x\times n_u})^3\times\mathbb{R}^{n_x\times n_x}
\end{align*}
such that, for all $\ell=1,\ldots,n_u$, and for any combination of $(i,j)$ in $\{0,1\}^2$, the following LMIs hold
\begin{equation}\label{eq:LMI_conditions1}
\begin{aligned}
   \Phi_{i,j}(A, B)=\begin{bmatrix}
             \mathrm {He}(H)-W_{j} -\frac{\lambda\varepsilon}{\mu}B_wB_w^\top  &-AH^\top -i B Y & -i B S^\top \\
        \ast &(1-\mu)W_{i} \!-\!\mathcal{Q}_{i} & {i}(Y+Z)^\top \\ \ast & \ast & 2iS
    \end{bmatrix}\succeq0,
\end{aligned}
\end{equation}
\begin{equation}\label{eq:LMI_conditions2}
\begin{aligned}
 \Psi_{i\ell }=\begin{bmatrix} W_{i} - \mathcal{Q}_{i} & (Z_{i})^{\top}_\ell \\\ast &\bar{u}^2\end{bmatrix}\succeq 0,
\end{aligned}
\end{equation}
where we called $\mathcal Q_{i}=(1 - 2i) Q$. 

Then, the switched control law \eqref{def:u} and the event-triggered rule \eqref{def:trig_indexed} with $K=Y{H^{-1}}^{\top}$ and $M= H^{-1} Q {H^{-1}}^{\top}$, respectively, allow satisfying, for $P_i=H^{-1}W_i{H^{-1}}^{\top}$, the following statements:
\begin{enumerate}[label=C\arabic*),ref=C\arabic*]
    \item \label{obj:basin} When $w=0$, $\mathcal A_1 := \left\{x\in\mathbb R^{n_x}~|~x^{\! \top}\!  P_{\sigma(x)} x\leq 1\right\},$ is an approximation of the basin of attraction of the origin for the closed-loop system \eqref{def:DSsys0}-\eqref{def:u}.
    \item \label{obj:attractor} When $w\in \Omega_\lambda$ and $w\neq0$, the solutions to the closed-loop system \eqref{def:DSsys0}-\eqref{def:u} converge to the attractor $\mathcal A_\varepsilon := \left\{x\in\mathbb R^{n_x}~|~x^{\! \top}\!  P_{\sigma(x)} x\leq \varepsilon^{-1}\right\}$, initialized in $\mathcal A_1\setminus\mathcal A_\varepsilon$.
\end{enumerate}
\end{theorem}

\begin{proof}
The proof consists of three steps. The first one is based on proving the positive definiteness of the candidate Lyapunov function. In the second one, the forward increment of the Lyapunov function along the solutions of system \eqref{def:DSsys0} is addressed, ensuring the asymptotic convergence of the trajectories to attractors $\mathcal A_1$ and $\mathcal A_\varepsilon$. Finally, the invariance of the attractor $\mathcal A_\varepsilon$ when $w\in\Omega_{\lambda}$ and $w\neq 0$ is proved in the third part.

\underline{\textit{Step 1:} Positivity of $V$.} Consider the Lyapunov function given by
\begin{equation}\label{def:Lyap}
   V(x)= x^{\! \top}\!  P_{\sigma(x)} x, \quad \forall x\in\mathbb R^{n_x},
\end{equation}
where matrices $\{P_i\}_{i\in\{0,1\}}\in\mathbb S^{n_x}$ are not necessarily definite positive and are to be defined. The positive definiteness of $V(x)$ is established by noting that
\begin{align*}
V(x)&=x^\top P_{\sigma(x)} x= x^\top \big(P_{\sigma(x)} -\!(1 - 2\sigma(x))M\big) x + \!(1 - 2\sigma(x))x^\top Mx.    
\end{align*}

Recalling that Lemma \ref{lem:M2} guarantees that $(1 - 2\sigma(x))\, x^\top Mx \geq0$, this implies
\begin{equation*}
    V(x)\geq x^\top \big(P_{\sigma(x)} -\!(1 - 2\sigma(x)) M\big) x.
\end{equation*} 

The assumption on $\Psi_{i,\ell}(A,B)\succ0$ implies $W_{i} -(1 - 2iQ)\succ0$ for any $i\in \{0,1\}$. Thus, inequality $W_{\sigma(x)} -(1 - 2\sigma(x))Q)\succ0$ holds. Pre- and post-multiplying by $H^{-1}$ and its transpose, respectively, the last inequality becomes $(P_{\sigma(x)} -\!(1 - 2\sigma(x))M)\succ0$, which proves that $V(x)>0$. 

\underline{\textit{Step 2:} Convergence from $\mathcal A_1$ to $\mathcal A_\varepsilon$}. From now on, we will use the short-hand notation $\sigma=\sigma(x)$ and $\sigma^+=\sigma(x^+)$ along the proof, for the sake of simplicity. 
To accomplish statements \ref{obj:basin} and \ref{obj:attractor}, the forward increment of the Lyapunov function along the trajectories of system \eqref{def:DSsys}, expressed as $\Delta V(x)= {x^+}^{\! \top}  P_{\sigma^+} x^+ - x^{\! \top}  P_{\sigma} x$, requires verification of
\begin{equation}\label{objectives_equations}
    \Delta V(x)\leq 0,\quad\forall(x,w)\,\text{s.t.}
    \begin{cases}
        x\in\mathcal A_1, & (\text{i.e.}, \, x^{\top}P_{\sigma}x\leq 1),\\
        x\notin\mathcal A_\varepsilon, & (\text{i.e.}, \, x^{\top}P_{\sigma}x\geq \varepsilon^{-1}),\\
        w \in \Omega_{\lambda}, & (\text{i.e.}, \, w^{\top}w\leq \lambda),\\
        \mathcal L_2(x)\geq0 & (\text{i.e.}, \,  (1-2\sigma)x^\top Mx \geq 0).
    \end{cases}
\end{equation}

To do so, let us introduce the following quantity
\begin{equation*}
\begin{split}
\mathcal{L}(x,w)=\Delta V(x)+\underbrace{\mathcal L_2(x)}_{\geq 0 \mbox{ \eqref{lem:condM_indexed}} }+\bar{\mu}_1\underbrace{(x^{\top}P_\sigma x-\varepsilon^{-1})}_{\displaystyle\geq 0 \mbox{ if } x\notin\mathcal A_\varepsilon}+ \bar{\mu}_2\underbrace{(\lambda-w^{\top}w)}_{\displaystyle\geq 0 \mbox{ if } w \in \Omega_{\lambda}}-\ \mathcal L_1(x)
    \end{split}
\end{equation*}
where $\bar\mu_i\geq 0,~i=1,2$ are some positive  scalars to be defined. 

Let us first show that constraining $x$ to be in $\mathcal A_1$ implies that $\mathcal L_1(x)\leq 0$. To do so, consider condition $\Psi_{\sigma,\ell}\succeq 0$ in \eqref{eq:LMI_conditions2}. Pre- and post-multiplying by $\left[\begin{smallmatrix}
    H^{-1} &0 \\0& 1
\end{smallmatrix}\right]$ and its transpose, respectively, yields
\begin{equation*}
\begin{bmatrix}
P_{\sigma}\!-\!\mathcal{M}_{\sigma} & (H^{-1}Z^{\top}_{{\sigma}})_\ell\\(Z_{{\sigma}}H^{-\top})_\ell & \bar u^2_\ell
\end{bmatrix}\succeq 0, \quad \forall \ell=1,2,3.
\end{equation*}

Introducing matrices $G_\sigma =H^{-1}Z_\sigma$ for $\sigma=0,1$ and applying the Schur complement leads to the following inequality
\begin{equation*}\label{sat_relationship}
P_{\sigma}-\mathcal M_\sigma - \frac{1}{\bar u^2_\ell}x^{\top}(G^{\top}_{{\sigma}})_\ell(G_{{\sigma}})_\ell \succeq 0,  \quad \forall \ell = 1,2,3.
\end{equation*}

Therefore, for any $x$ in $\mathcal A_1$, we get 
\begin{equation*}
 \frac{1}{\bar u^2_\ell}x^{\top}(G^{\top}_{{\sigma}})_\ell(G_{{\sigma}})_\ell x\leq x^{\top}(P_{\sigma}-\mathcal M_\sigma) x=x^{\top}P_{\sigma}x - \underbrace{x^\top ((1-2\sigma)M)x}_{=\mathcal L_2(x)\geq0} \leq \underbrace{V(x) \leq 1 }_{x\in\mathcal A_1}, \quad \forall \ell = 1,2,3.
\end{equation*}

Hence, the feasibility of condition $\Psi_{\sigma,\ell}\succ0$ is ensured, provided that $\mathcal A_1\subset\mathcal{S}(G)$. Recalling  that $\sat(\sigma Kx)=\phi(\sigma Kx)+\sigma Kx$ and that $\phi(\sigma Kx)=\sigma\phi(Kx)$, due to the fact that $\sigma=0$ or $1$, Lemma~\ref{lem:M1} guarantees that
\begin{equation*}
\mathcal{L}_1(x) =2\sigma\phi( Kx)^\top T [  \sigma\phi(Kx)+\sigma Kx + G_{\sigma}x] \leq 0.
\end{equation*}


Hence, if one can guarantee that $\mathcal L(x,w)$ is negative definite, it necessarily implies the satisfaction of \eqref{objectives_equations}, where $\Delta V(x)\leq0$ under the particular constraints on $x$. This is due to the fact that in this case we would have $\Delta V(x)\leq \mathcal L(x,w)\leq 0$. Then, to prove that $\mathcal L(x,w)\leq0$, let us first introduce the augmented vector $\xi=\begin{bmatrix}{x^+}^{\top} & x^{\top} & \phi^{\top}(Kx) & w^{\top}\end{bmatrix}^{\top}$, in order to have 
\begin{equation}\label{eq:main_diseq}
   \mathcal{L}(x,w)= \xi^{\top} \begin{bmatrix} 
   P_{\sigma^{+}} &0 & 0 &0\\
        \ast &(-1+\bar{\mu}_1)P_{\sigma}\!+\!\mathcal{M}_{\sigma} & -\sigma( \sigma K+G_{\sigma})^{\top}T &0\\
        \ast & \ast & -2\sigma^2 T &0 \\ \ast &0 &0 &-\bar{\mu}_2I
    \end{bmatrix}\xi-\bar{\mu}_1\varepsilon^{-1}+\bar{\mu}_2\lambda.
\end{equation}

Then, noticing that $\sigma^2=\sigma\in\{0,1\}$ and selecting $\bar\mu_1=\mu\in(0,1)$ and $\bar\mu_2=\mu/(\varepsilon\lambda)$ yields 
\begin{equation}\label{eq:main_diseq2}
   \mathcal{L}(x,w)= \xi^{\top} \begin{bmatrix} 
   P_{\sigma^{+}} &0 & 0 &0\\
        \ast &-(1-\mu)P_{\sigma}\!+\!\mathcal{M}_{\sigma} & -\sigma( K+G_{\sigma})^{\top}T &0\\
        \ast & \ast & -2\sigma T &0 \\ \ast &0 &0 &-\frac{\mu}{\varepsilon \lambda} I
    \end{bmatrix}\xi.
\end{equation}

In addition, it follows from the expression of the discrete-time system \eqref{def:DSsys} that the following identity links the elements of $\xi$:
\begin{equation}\label{eq:finslercond}
   \xi^\top \begin{bmatrix}
        H^{-1}\\0\\0\\0
    \end{bmatrix} \begin{bmatrix}
        -I & A+\sigma B K & \sigma B & B_w
    \end{bmatrix}\xi=0, \quad \forall \xi\in\mathbb R^{3n_x+n_u},
\end{equation}
where we used the inverse of matrix $H$, solution to the LMI conditions. This yields 
\begin{equation*}
 \mathcal L(x,w)=\xi^\top  \begin{bmatrix} 
   P_{\sigma^{+}}-\mathrm {He}(H^{-1}) &H^{-1}(A+\sigma BK)& \sigma H^{-1}B & H^{-1}B_w\\
        \ast &-(1-\mu)P_{\sigma}\!+\!\mathcal{M}_{\sigma} & -\sigma( K+G_{\sigma})^{\top}T &0\\
        \ast & \ast & -2\sigma T &0 \\ \ast &0 &0 &-\frac{\mu}{\varepsilon \lambda}I
    \end{bmatrix}\xi.
\end{equation*}
 
Defining $S=T^{-1}$, and pre- and post-multiplying the previous matrix inequality by $\text{diag}(H,H,S,I)$ and its transpose, respectively, 
and selecting $Q=H MH^\top$, $Y=KH^\top$, $Z=GH^{\top}$ and $W_{i}=H P_{i}H^\top,~i\in\{0,1\}$, we arrive at 
   \begin{equation*}
  \mathcal L(x,w)\leq 0 \quad \Leftarrow \quad   
    \left[ 
    \begin{matrix}
             W_{\sigma^{+}} -\mathrm {He}(H) &AH^\top +\sigma BY & \sigma BS^\top&B_w \\
        \ast &-(1-\mu)W_{\sigma} \!+\!\mathcal{Q}_{\sigma} & -\sigma (Y+Z_{\sigma})^\top &0 \\ \ast & \ast & -2\sigma S &0 \\\ast &0 &0 &-\frac{\mu}{\varepsilon \lambda} I
    \end{matrix}\right]\preceq 0.
\end{equation*}
\indent The Schur complement is applied to the previous inequality, yielding, with a change of sign, condition $\Phi_{\sigma,\sigma^+}(A,B)\succeq 0$.
Consequently, inequality $\mathcal{L}(x,w)\leq 0$ is guaranteed if LMIs $\Phi_{\sigma,\sigma^+}(A,B)\succeq 0$ hold for all pairs $(\sigma,\sigma^+)$ in $\{0,1\}$, as in \eqref{eq:LMI_conditions1}. 
Therefore, conditions \eqref{eq:LMI_conditions1} and \eqref{eq:LMI_conditions2} ensure the local asymptotical stability of the closed-loop system \eqref{def:DSsys0}-\eqref{def:u} with $K=Y{H^{-1}}^{\top}$, for any initial condition $x_0\in\mathcal A_1$. 


\underline{\textit{Step 3:} Invariance.} It only remains to prove that the attractor $\mathcal A_\varepsilon$ is invariant. On the one hand, given $\varepsilon>1$ solution to conditions \eqref{eq:LMI_conditions1} and \eqref{eq:LMI_conditions2}, we get $\mathcal A_\varepsilon\subseteq\mathcal A_1$ and the convergence of the system trajectories to the set $\mathcal A_\varepsilon$. On the other hand, the satisfaction of $\Phi(A,B)\succ 0$ ensures that $\mathcal{L}(x,w)$ in \eqref{eq:main_diseq2} is negative semi-definite, yielding
\begin{align*}
V(x^+)&= V(x)
  + \underbrace{\mathcal{L}(x,w)}_{\le 0}
  + \underbrace{\mathcal{L}_1(x,w)}_{\le 0}
  - \underbrace{\mathcal{L}_2(x,w)}_{\ge 0} - \bar{\mu}_1\!\left(V(x)-\varepsilon^{-1}\right)
      - \bar{\mu}_2 \underbrace{(\lambda-w^{\top}w)}_{\ge 0} \notag\\
&\le (1-\bar{\mu}_1)V(x)
   + \bar{\mu}_1\varepsilon^{-1}.
\end{align*}
\indent Since $\bar{\mu}_1=\mu\in(0,1)$ and  $V(x)\leq\varepsilon^{-1}$, we have
$V(x^{+})\leq \varepsilon^{-1}$, which implies that $x^{+}\in\mathcal A_\varepsilon$, completing the proof.
\end{proof}

\begin{remark}
    Theorem \ref{th:Stabb} contains a simple but not trivial feature. The Lyapunov function that results from the provided analysis is positive definite, whereas the Lyapunov matrices $P_\sigma$ are not required to have the same property. This means that the individual functions $V(x)=x^\top P_ix$ for $i\in\{0,1\}$ are not necessarily positive definite. However, the positive definiteness is imposed on a specific partition of the state space. More specifically, we have 
    \begin{align}
        V(x)=x^\top P_0x>0,\quad \forall x\in\mathbb R^{n_x}\quad \mbox{ s.t. }\quad x^\top Mx\geq0,\nonumber\\
        V(x)=x^\top P_1x>0,\quad \forall x\in\mathbb R^{n_x}\quad \mbox{ s.t. }\quad x^\top Mx\leq0.\nonumber
    \end{align}
\end{remark}


\subsection{Discussion on Theorem \ref{th:Stabb} and optimization procedure}
Theorem~\ref{th:Stabb} provides feasibility conditions for the design of a stabilizing event-triggered controller. However, at this stage, no guarantees are provided regarding the optimality of the resulting control law with respect to any specific performance criterion. This section introduces several optimization frameworks that can be readily incorporated into the solution of Theorem~\ref{th:Stabb} within the LMI-based approach. Prior to delving into the details of these approaches, we first present a graphical interpretation of the various components of Theorem~\ref{th:Stabb}, which serves to clarify the rationale behind the proposed optimization formulation.
\subsubsection{Graphical interpretation of Theorem \ref{th:Stabb}}\label{sec:graphical}

To facilitate a deeper understanding of the main result, we provide a brief discussion on the interpretations and implications of Theorem~\ref{th:Stabb}. Figure~\ref{ellipsoids} offers a graphical representation of the key components of the theorem. Specifically, the figure illustrates the manifold $x^\top M x = 0$, which delineates the boundary between the regions where $\sigma = 0$ (top-left) and $\sigma = 1$ (bottom-right). Within each of these regions, the Lyapunov function $V(x) = x^\top P_\sigma x$ defines two level sets: $V(x) = 1$, which serves as an estimate of the basin of attraction $\mathcal{A}_1$, and $V(x) = \varepsilon^{-1}$, corresponding to the attractor $\mathcal{A}_\varepsilon$. 
\begin{figure}[t]
\centering
\includegraphics[width=0.45\textwidth]{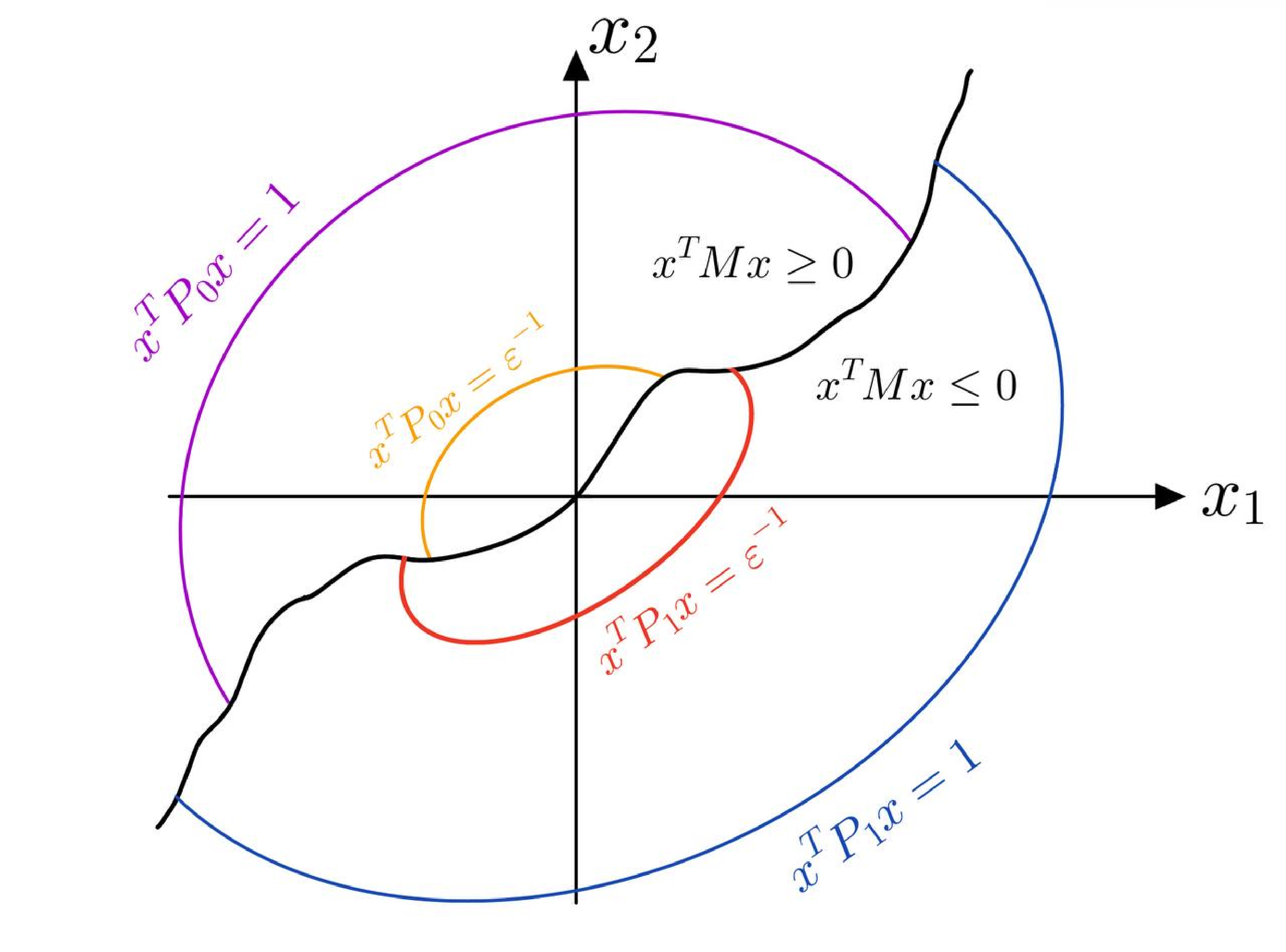}
\caption{Attractor and basin of attraction (plane projection for any two generic states, $x_1$ and $x_2$)}
\label{ellipsoids}
\end{figure}\\ \indent
In summary, the main contribution of Theorem~\ref{th:Stabb} is as follows. The positive definiteness of the Lyapunov function $V$ ensures that the set $\mathcal{A}_1$ is closed and bounded, thereby defining a well-defined region of admissible initial conditions. Statements $C_1$ and $C_2$ in Theorem~\ref{th:Stabb} then guarantee that, for any initial condition within $\mathcal{A}_1$, the trajectories of the event-triggered closed-loop system converge to the set $\mathcal{A}_\varepsilon$. Moreover, the invariance property ensures that once a trajectory enters $\mathcal{A}_\varepsilon$, it remains within this set for all subsequent times. 

\subsubsection{Optimization of the basin of attraction}
Due to the presence of input saturation, stability properties are generally no longer global, particularly for open-loop systems that are not asymptotically stable, as is the case for the HCW system, which is only marginally stable. In this context, a natural objective is to enlarge the approximation of the basin of attraction, i.e., the set of admissible initial conditions guaranteeing asymptotic stabilization. In Theorem~\ref{th:Stabb}, this estimation is given by the set $\mathcal A_1=\{x\in\mathbb R^{n_x},\mbox{ s.t. } x^\top P_\sigma(x)x\leq 1\}$. Hence, the size of $\mathcal{A}_1$ can be increased by minimizing the eigenvalues of the matrices $P_0$ and $P_1$. The corresponding optimization problem is formulated as follows.
\begin{opt} Optimization of the basin of attraction:
    \begin{equation}
\begin{aligned}
\min_{\mathcal{D}\in \mathbb H^{n_x},\eta >0 } \ &\eta\\
\textrm{s.t.} \ &\Phi_{i,j}(A,B)\!\succeq\!0,\ \Psi_{i,\ell}\!\succeq\!0,\ \begin{bmatrix}\eta I & \!\!\!I\\I &\!\!\!\mathrm {He}(H)\!-\!W_{i}\!+\! \mathcal{Q}_{i}\end{bmatrix}\!\succeq\! 0,\ \forall (i,j)=\{0,1\}^2,\ \forall \ell =1,\dots, n_u.
\end{aligned}
\end{equation}
\end{opt}
\begin{proof}
    Using the inequality $(W_i-\mathcal Q_i-H)^\top(W_i-\mathcal Q_i)^{-1}(W_i-\mathcal Q_i-H)\succeq 0$, which holds true because of the assumption on $ W_i-\mathcal Q_i\succ0$ implied by the constraint on $\Psi_{i,\ell}$, we get that $H^\top (W_i-\mathcal Q_i)^{-1}H \succeq\mathrm {He}(H)-(W_{i}- \mathcal{Q}_{i}) $. This yields
    $$
    \begin{bmatrix}\eta I & I\\I &H^\top (W_i-\mathcal Q_i)^{-1}H\end{bmatrix}\succeq \begin{bmatrix}\eta I & I\\I &\mathrm {He}(H)-W_{i}+ \mathcal{Q}_{i}\end{bmatrix}\succeq 0.
    $$

Applying the Schur complement to the first matrix on the left guarantees that $\eta I\succeq H^{-1} (W_i-\mathcal Q_i)H^{\top}=P_i-\mathcal M_i=P_i-(1-2i)M_i$. Therefore, minimizing $\eta>0$ implies the minimization of the eigenvalues of matrices $P_i$.
\end{proof}


 

\subsubsection{Optimization of the attractor}
The dynamics of system \eqref{def:DSsys} is subject to external disturbances $w$, which satisfies Assumption \ref{th:bounded_noise}. This implies that the closed-loop trajectories cannot converge asymptotically to the origin, but only to a neighborhood surrounding it, i.e. the attractor $\mathcal A_\varepsilon$. Accordingly, an additional optimization objective consists in minimizing the size of this attractor. As the attractor is characterized by the decision variable $\varepsilon>1$, we derive the following optimization problem:
\begin{opt} Optimization of the attractor:
\begin{equation}
\begin{aligned}
\max_{\mathcal{D}\in \mathbb H^{n_x}} \quad &\varepsilon\\
\textrm{s.t.} \quad &\Phi_{i,j}(A,B)\succeq 0,\ \Psi_{i,\ell}\succeq 0,\ \forall (i,j)=\{0,1\}^2,\ \forall \ell =1,\dots, n_u.
\end{aligned}
\end{equation}
\end{opt}


\subsubsection{Event-triggering rule}
To encourage switching between the two conditions, it is desirable for the matrix $M$ to be neither positive definite nor negative definite. Equivalently, this requirement can be expressed by imposing that $M$ possesses both strictly positive and strictly negative eigenvalues, which can be enforced by setting its trace to zero. This condition can be translated to the decision variable $Q$ in the optimization problem, leading to the constraint $\text{trace}(Q) = 0$. However, due to the presence of the term $\mathcal{W}_{\sigma} - \mathcal{Q}_{\sigma}$ in the LMIs, where the involved matrices have opposite signs, the optimization process naturally tends to drive the optimization to look for $\mathcal{Q}_{\sigma}\approx 0$. To ensure that both $Q$ and, consequently, $M$ remain nonzero, we propose to promote their magnitude by maximizing $\text{trace}(Q)$. It can be noticed that this feature further increases the probability that the triggering condition $x^{\top} M x$ lies in its positive region. Recalling that this corresponds to $\sigma = 0$, the system consequently spends more time in the region where no control action is applied, reducing the overall number of thruster firings.
This leads to the following optimization problem: 
\begin{opt} Reduction of the number of firing instants:
\begin{equation}
\begin{aligned}
\max_{\mathcal{D}\in \mathbb H^{n_x}} \quad &\text{trace}(Q) \\
\textrm{s.t.} \quad &\Phi_{i,j}(A,B)\succeq 0,\ \Psi_{i,\ell}\succeq 0,\ \forall (i,j)=\{0,1\}^2,\ \forall \ell =1,\dots, n_u.
\end{aligned}
\end{equation}
\end{opt}
\indent Reducing the number of firings is inherently beneficial, as each valve actuation consumes part of the thrusters lifetime. Triggering control only when necessary extends hardware life, allowing longer missions or contingency margins. Fewer firings also lessen induced torques, reducing attitude-control energy. Finally, minimizing burns decreases mode transitions, fault-protection activations, and thermal cycling, therefore simplifying fault management, lowering anomaly risk, and improving overall mission reliability.
Clearly, for a fixed total energy expenditure, a reduction in the number of firing instants implies that each individual firing must deliver a higher average impulse. This aspect, however, is not explicitly addressed within the current optimization framework. In other words, while our method effectively minimizes the frequency of firings, it does not independently control how the total impulse (and thus the energy) is distributed among them.
Future work should therefore focus on decoupling these two objectives, minimizing the number of firing instants and minimizing the total energy cost, so that they can be analyzed and optimized separately. It is indeed worth mentioning that the LMI framework offers the possibility of including, at a very low
off-line computational cost, the design of an optimized controller to decrease the total energy.
Incorporating this feature in future developments would enable the design of control strategies that explicitly balance the trade-off between the frequency and intensity of thruster firings.

\subsubsection{Tradeoff optimization}

In light of these discussions, it is worth noting that a trade-off can be achieved by formulating a combined optimization problem, defined for any non-negative weighting parameters $\alpha_1 \geq 0$, $\alpha_2 \geq 0$, and $\alpha_3 \geq 0$, as follows:

\begin{opt} Mixed optimization:
    \begin{equation}
\begin{aligned}
\min_{\mathcal{D}\in \mathbb H^{n_x},\eta >0 } \  &\alpha_1 \eta - \alpha_2 \varepsilon+\alpha_3\text{trace}(Q)\\
\textrm{s.t.} \  &\Phi_{i,j}(A,B)\succeq 0,\ \Psi_{i,\ell}\succeq 0,\ \begin{bmatrix}\eta I & \!\!\!I\\I &\!\!\!\mathrm {He}(H)\!-\!W_{i}\!+\! \mathcal{Q}_{i}\end{bmatrix}\!\succeq \! 0,\ \forall (i,j)=\{0,1\}^2,\ \forall \ell =1,\dots, n_u.
\end{aligned}
\end{equation}
\end{opt}

\section{Numerical application}\label{sec:simulations}
Simulations are performed in \texttt{Matlab} after solving the optimization problem from Theorem \ref{th:Stabb} once by using \texttt{cvx solver}. From its solution, it is possible to compute the gain matrix $K$ and the switching condition matrix $M$. The parameters related to the noise handling and to the maximization function are selected as $\lambda=10^{-7}$, $\mu=0.04$, $\alpha_1=1$, $\alpha_2=100$, and $\alpha_3=1$. The proposed control law is validated through comparison with an MPC scheme formulated in a simple formulation, considering a quadratic cost and input saturation constraints. This choice provides a straightforward baseline for comparison; however, it is worth emphasizing that more sophisticated MPC frameworks for spacecraft rendezvous, such as the one described in \citep{rebollo2024mpctrackingappliedrendezvous}, achieve stronger performance and theoretical guarantees on feasibility and stability. More precisely, the MPC scheme used in this paper is based on the solution, at each sampled instant $t_k$ of the simulation, of the following minimization problem:
 \begin{equation}
  \begin{array}{rl}  \displaystyle \mathop{min}_{u_F(t_k)\in([-\bar u, \bar u]^{n_u})^{N_p}
  } &
    \frac{1}{2}u_F(t_k)^{\top} \tilde{H}u_F(t_k)+\tilde{f}^{\top}(x(t_k))u_F(t_k)\\
    \mbox{s.t } & x_k=\hat{x}(t_k),\quad x_{k+j+1}=Ax_{k+j}+Bu_{k+j}, \qquad  \text{for} \ j=0,...N_p-1
    \label{eq:eq3.16}
    \end{array}
    \end{equation}
where $u_F(t_k)$ is the augmented vector made of all the future control inputs, starting from $u(t_k)$,  over the considered prediction horizon, of length $N_p$.  $\hat{x}(t_k)$ represents the initial condition computed as the last instant of the previous iteration.
$\tilde{H}$ and $\tilde{f}(x_k)$ are properly selected weighting matrices for the minimization function. In particular, $\tilde{H}=\tilde{G}_u^{\top}\tilde{Q}_f\tilde{G}_u+\tilde{R}_f+\tilde{G}_{u_N}^{\top}\tilde{P}\tilde{G}_{u_N}$ and $\tilde{f}(x_k)=2x_k(\tilde{F}^{\top}\tilde{Q}_f\tilde{G}_u+\tilde{F}_{N}\tilde{P}\tilde{G}_{u_N})$, where $
(\tilde{F})_{ij} = A^i$ for all $i=1\dots, N_p$ if $j = 1$ and $0$ otherwise, $(\tilde{G}_u)_{ij} = A^{i-j} B$ for all $i= 1 \dots  N_p$ and all $j=1 ,\dots,  N_p,$ such that $i \ge j$ and where $\tilde{Q}_f$ and $\tilde{R}_f$ are properly selected weight block-diagonal matrices. $\tilde{G}_{u_N}$ and $\tilde{F}_{N}$ are matrices made up of the last $n_x$ rows of $\tilde{G}_u$ and $\tilde{F}$, respectively, and $\tilde{P}=10P_{LQR}$, with $P_{LQR}$ solution to the discrete-time algebraic Riccati equation.

\subsection{Linear simulation model with additive disturbances}
The control law is first validated by using the same propagation and disturbance model used in the control law: the HCW equations and the additive noise, respectively. Only by using this setup and by selecting initial conditions inside the basin of attraction given by $\mathcal{A}_1$, we have a guarantee on the stability of our control law. The simulation time and sampling time are chosen respectively as \( T_{\text{sim}} = 40\,\text{min} \) and \( T = 10\,\text{s} \), resulting in $N_{sim}=240$ control intervals. The orbit is selected with a radius of $R_0=6878\,\text{km}$, corresponding to a mean motion $n = 0.0011 \,\text{rad/s}$. The magnitude of the maximum input velocity for a single thruster is  $\bar{u} = 0.2\, \ m/s$. Regarding the MPC, the selected prediction horizon is $N_p=N_{sim}/2$. Three separate initial conditions are evaluated, all of them internal to the basin of attraction (meaning that for all of them $x_0^{\top}P_{\sigma(x_0)}x_0\leq1 $ holds). The control objective is to stabilize the chaser at the target state. Nevertheless, because of the effect of noise, as discussed in Section~\ref{sec:graphical}, convergence can only be guaranteed to a neighborhood of the target, characterized by the attractor $x^{\top}P_{\sigma(x)}x\leq\varepsilon^{-1}$. For each initial condition, Montecarlo simulations with 100 random realizations are performed. The values of the statistical properties of the white noise $w$ are the same used inside the optimization problem. The sets of initial conditions are:







\begin{table}[H]
    \centering
    
    \begin{tabular}{|c|c|c|c|c|c|c|}
        \cline{2-7}
        \multicolumn{1}{c|}{}& $r_x\, \ [m]$ & $r_y\, \ [m]$ & $r_z\, \ [m]$ & $v_x\, \ [m/s]$ & $v_y\, \ [m/s]$ & $v_z\, \ [m/s]$ \\
        \hline
        Case 1 & -180 & 220 & -100 & -0.1 & 0.15 & 0.15 \\
        Case 2 & 50 & 170 & -140 & 0.15 & -0.1 & -0.15\\
        Case 3 & -30 & 200 & 170 & 0.1 & -0.15 & 0.1 \\
        \hline
    \end{tabular}
      \caption{Initial conditions for three simulation runs}
  \label{tab:initial_conditions}
\end{table}

A 3D plot for the  superimposition of 10 random realizations for each one of the three cases, is represented in Fig. \ref{fig_chaser_3D_trajectory}.


\begin{figure}[H]
    \centering
    \begin{minipage}[b]{0.48\textwidth}
        \centering
        \includegraphics[width=\textwidth]{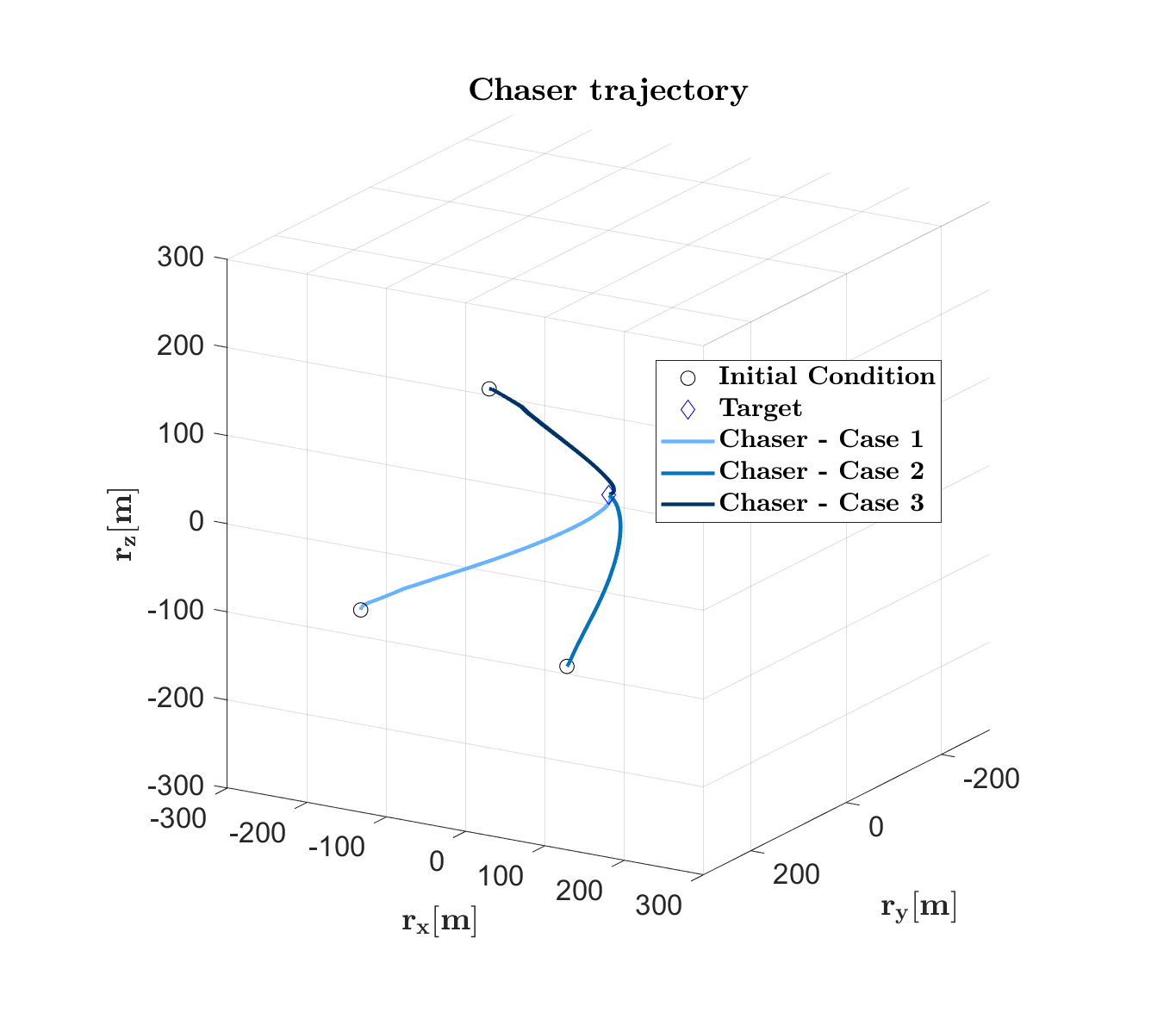}
        \caption{Chaser trajectories for three different initial conditions (10 random realizations each) using the ETC - Linear simulation model with additive disturbances}
        \label{fig_chaser_3D_trajectory}
    \end{minipage}
    \hfill
    \begin{minipage}[b]{0.48\textwidth}
        \centering
        \includegraphics[width=\textwidth]{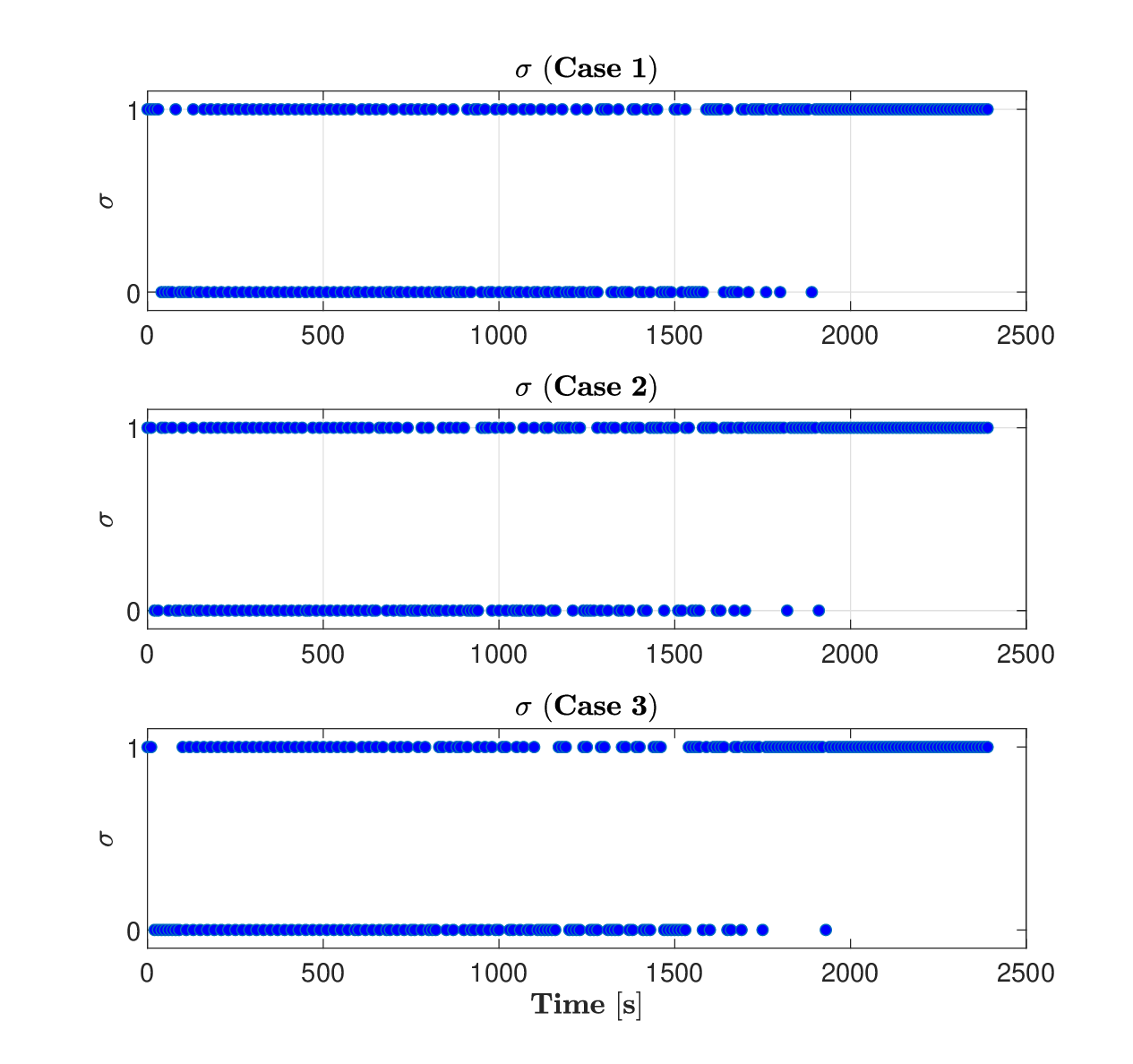}
        \caption{Values of the $\sigma$ parameter for three different initial conditions (single random realization) using the ETC - Linear simulation model with additive disturbances}
        \label{fig_sigma_3simulations}
    \end{minipage}
\end{figure}

As shown in Figure \ref{fig_chaser_3D_trajectory}, which represents the superimposition of 10 random realizations for each one of the initial conditions in Table~\ref{tab:initial_conditions}, the chaser is always able to achieve its goal and reach the target location. To better appreciate the effectiveness of our switching law, in Fig. \ref{fig_sigma_3simulations}, the values assumed by the $\sigma$ parameter are shown for a single realization of each of the three cases of initial condition. The switching law proves effective in alternating between the partitions of the state-space that require or not a control action. 

A comparison with MPC, with prediction horizon $N_p=N_{sim}/2$ and a single random realization for Case 1 in Table~\ref{tab:initial_conditions}, in terms of relative positions and velocities, is provided in Figures \ref{fig_r_ETCvsMPC} and \ref{fig_v_ETCvsMPC}, respectively. The alternation between red and green segments for the ETC solution evidences the switching between the two possible modes: red intervals represent the natural motion following the application of an \textcolor{red}{?}  (CL standing therefore for closed-loop); green ones indicate instead natural motion without an initial impulse application (with OL standing for open-loop). The gray and dashed trajectories represent the MPC solution. For both control strategies, the control objectives are achieved in the chosen simulation time. A plot of the control inputs history of the ETC is depicted in Fig. \ref{fig_control_inputs}.

\begin{figure}[t]
    \centering
    \begin{minipage}[b]{0.48\textwidth}
        \centering
        \includegraphics[width=\textwidth]{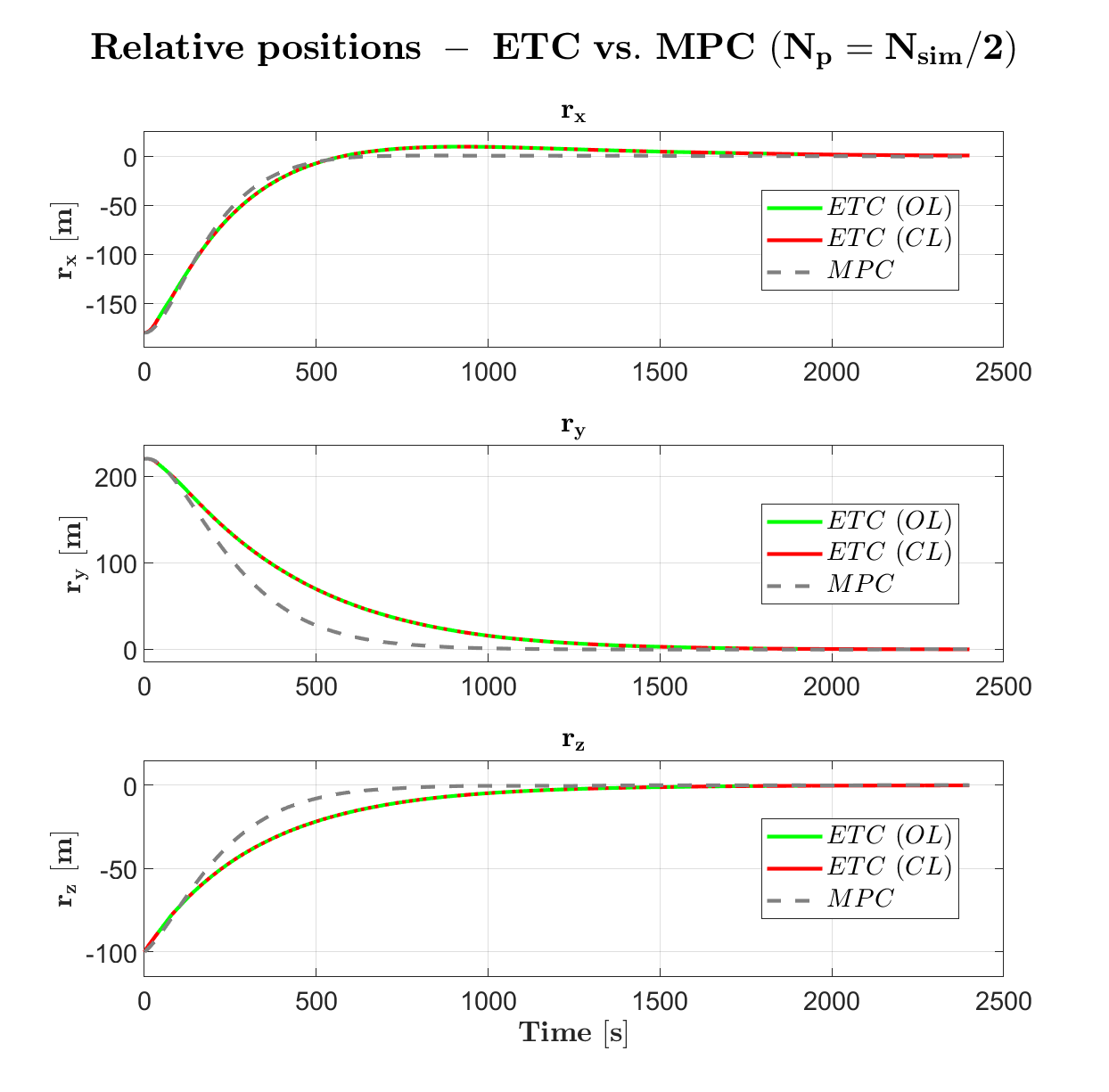}
        \caption{Evolutions of the chaser relative positions using the ETC (red: intervals following the application of an impulse; green: intervals without an initial impulse application) and MPC (gray dashed lines) strategies - Case 1 (single random realization) - Linear simulation model with additive disturbances}
        \label{fig_r_ETCvsMPC}
    \end{minipage}
    \hfill
    \begin{minipage}[b]{0.48\textwidth}
        \centering
        \includegraphics[width=\textwidth]{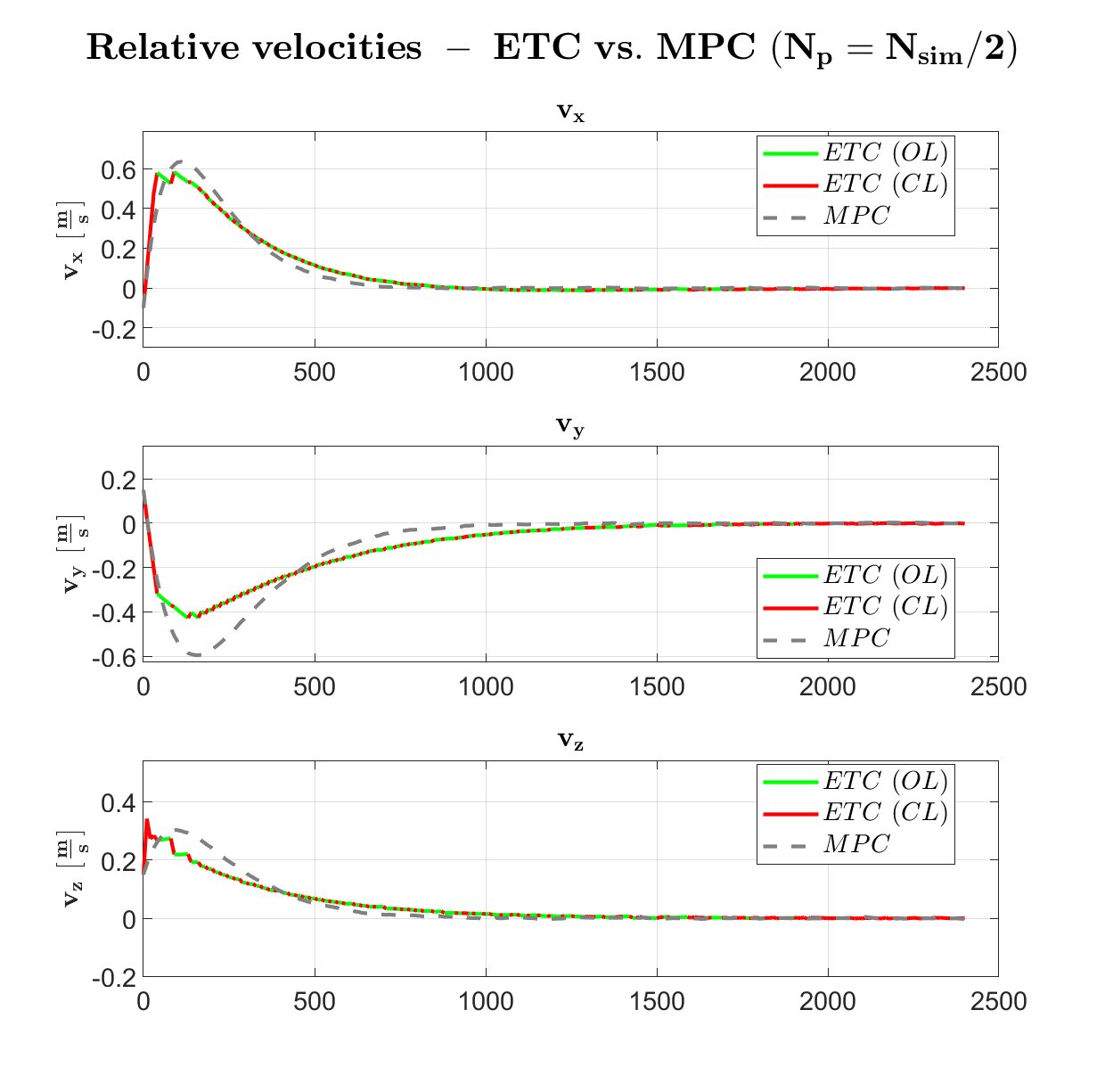}
        \caption{Evolutions of the chaser relative velocities using the ETC (red: intervals following the application of an impulse; green: intervals without an initial impulse application) and MPC (gray dashed lines) strategies - Case 1 (single random realization) - Linear simulation model with additive disturbances}
        \label{fig_v_ETCvsMPC}
    \end{minipage}
\end{figure}

\begin{figure}[H]
\centering
\includegraphics[width=0.45\textwidth]{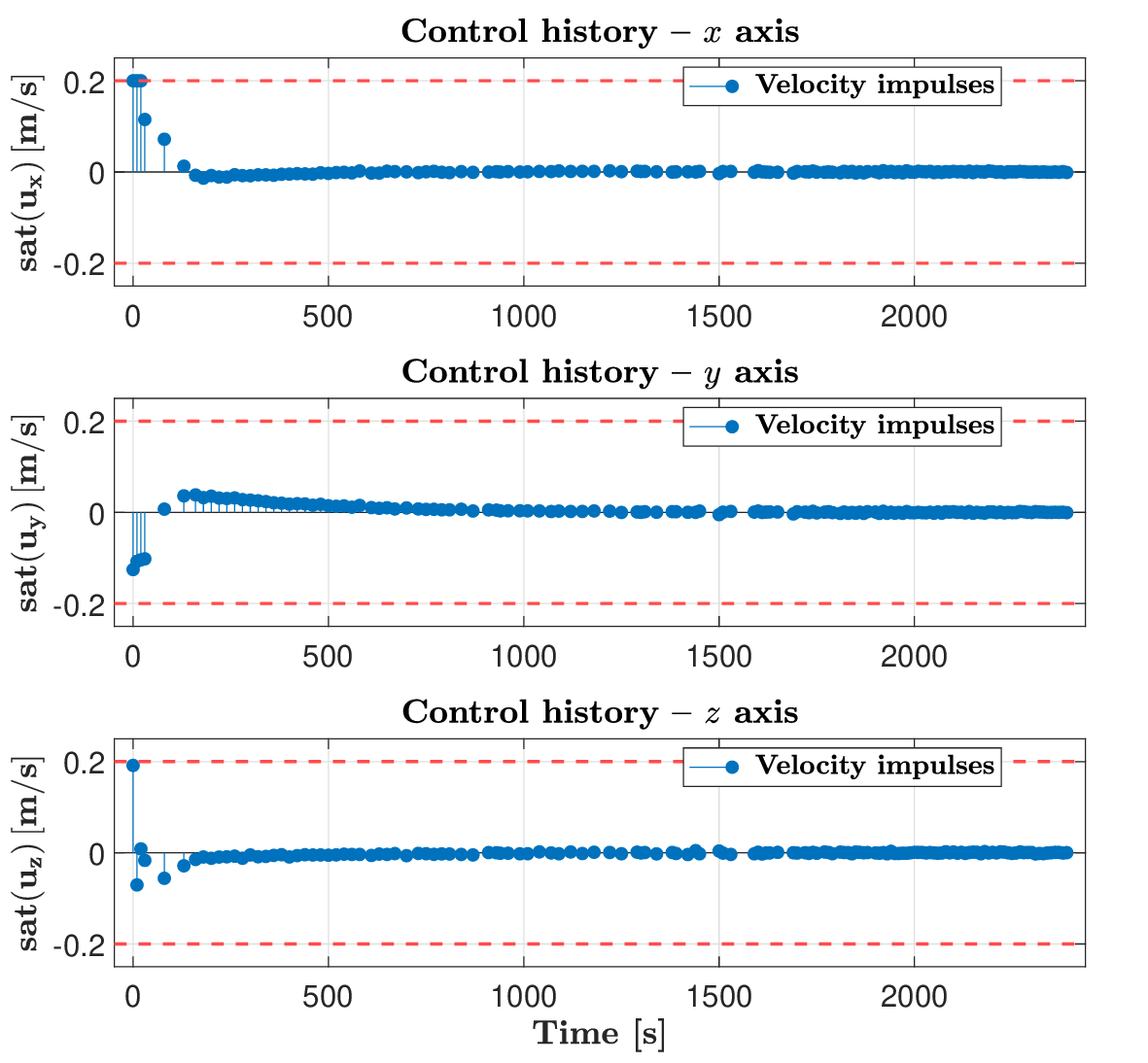}
\caption{Control history - Case 1 (single random realization) - Linear simulation model with additive disturbances}
\label{fig_control_inputs}
\end{figure}

To better put in comparison the two methods, the following table evidences the values of different control performances, averaged over all random realizations:

\begin{table}[H]
    \centering
    \begin{tabular}{c|cc|cc|} 
        \cline{2-5}
         & \multicolumn{2}{c|}{$\bar{u}_{tot}\, [\mathrm{m/s}]$} 
         & \multicolumn{2}{c|}{$\% \text{ of firing instants}$} \\
        \cline{2-5}
         & ETC & MPC\ & ETC & MPC\\
        \hline
        \multicolumn{1}{|c|}{Case 1} & 2.937 & 3.126 &  57.92 & 100 \\
        \multicolumn{1}{|c|}{Case 2} & 2.527 & 2.353 &  62.92 & 100 \\
        \multicolumn{1}{|c|}{Case 3} & 2.427 & 2.535 &  59.58 & 100 \\
        \hline
    \end{tabular}
        \caption{Performance indexes: ETC vs MPC}
 \label{tab:performance_indexes}
\end{table}


In Table~\ref{tab:performance_indexes}, $\bar{u}_{tot}$ is the total energy cost, given by the value of $u_{tot}=\sum_{k=1}^{N_{sim}}|\text{sat}(u_{x_k})|+|\text{sat}(u_{y_k})|+|\text{sat}(u_{z_k})|$, averaged over all random realizations. \\
\indent As shown in the table, although our optimization problem primarily targets a reduction in the number of firing instants rather than directly minimizing total energy expenditure, our approach also achieves a comparable total energy cost compared to MPC. Nevertheless, the considerable decrease in the number of thruster activations remains indeed the most significant improvement over MPC, since this method typically generates a control input at every sampling instant.\\
\indent The second major advantage of our approach clearly lies in the significantly reduced computational time: in MPC, $N_{sim}$ optimization problems must be solved in online using the limited on-board hardware of the chaser satellite. By contrast, in the ETC strategy, only a single optimization problem is solved, and this is done offline, allowing the use of powerful ground-based computational resources.






\subsection{Non-linear simulation model}
The robustness of the proposed control strategy is now validated against a more realistic case, where the propagation model of both spacecraft is based on the 2-Body Problem \eqref{2bp}. The dynamics of the target and chaser are propagated separately, and the relative state is reconstructed by subtraction and transformation into the LVLH frame to be evaluated by the switching law.\\
\indent The target orbit is selected as circular (eccentricity $e=0$), with an inclination $i=30$°, and with a right ascension of the ascending node (RAAN) $\Omega=45$°, again with a radius $R_0=6878$ km.\\
\indent To increase the realism of the scenario, disturbance sources beyond simple additive noise are considered. In particular, two types of perturbation are taken into account: imperfect thrust execution and sensor noise. The first arises from the fact that the actual thrust may differ from the commanded value due to imperfections in the construction of the thrusters and possible misalignments. This is modeled as a random disturbance added to the commanded thrust ($\delta u_i$), combined with a random rotation (assumed to involve small angles ($\delta \theta_i$)) applied to it, so that the real thrust contribution at time $t_k$ results in $u^{real}(x(t_k))=[\sat(\tilde{u}_x(x(t_k)))\,\,
            \sat(\tilde{u}_y(x(t_k)))\,\,\sat(\tilde{u}_z(x(t_k)))]^{\top}$, where
\begin{equation}\label{def:u_noise}
    \tilde{u}(x(t_k))=R_{\delta\theta}(\delta\theta(t_k))({\sigma}(x(t_k))( Kx(t_k)+\delta u(t_k))),
\end{equation}
and where $R_{\delta\theta}$ is a given rotation matrix for small angles.\\
\indent The sensor noise, on the other hand, is modeled as random noise affecting the measured position ($\delta r_i)$ and velocity ($\delta v_i$) at each sampling instant. Since this noisy measurement is used as input to the switching law to determine the appropriate control mode, such errors may lead the control law to select an incorrect strategy at a given time:

\begin{equation}
\label{def:trig_indexed_noise}
\sigma(x(t_k)) =
\begin{cases}
0, & \text{if }\  \big(x(t_k) + 
\delta x(t_k) \big)^\top\! M\big(x(t_k) + 
\delta x(t_k) \big) \geq 0 \\
1, & \text{if }\  \big(x(t_k) + 
\delta x(t_k) \big)^\top\! M \big(x(t_k) + 
\delta x(t_k) \big) \leq 0
\end{cases},
\end{equation}

with $\delta x=[\delta r\, \,\delta v]^{\top}$.
Table \ref{tab:disturbance_parameters} specifies the values of means and standard deviations of the random noise contributions, taking into account the increase of sensor accuracy in the last phase of the rendezvous (considered as when all the relative distances fall inside a range of 20 meters).

\begin{table}[H]
    \centering 
    \resizebox{\textwidth}{!}{%
    \begin{tabular}{|c|c|c|cc|cc|}
        \cline{4-7}
  \multicolumn{1}{c}{} & \multicolumn{1}{c}{} & \multicolumn{1}{c|}{} 
& \multicolumn{1}{c}{$r_i>20\,\mathrm{m}$} & \multicolumn{1}{c|}{$r_i<20\,\mathrm{m}$} & \multicolumn{1}{c}{$r_i>20\,\mathrm{m}$} & \multicolumn{1}{c|}{$r_i<20\,\mathrm{m}$}
\\
        \cline{2-7}
        \multicolumn{1}{c|}{} & $\delta u_i\, [\mathrm{m/s}]$ 
        & $\delta \theta_i\, [\mathrm{deg}]$ 
        & \multicolumn{2}{c|}{$\delta r_i\, [\mathrm{m}]$} 
        & \multicolumn{2}{c|}{$\delta v_i\, [\mathrm{m/s}]$} \\
        \hline
        Mean & $\bar{u}/1000$ & 1 & $10^{-1}$ & $10^{-2}$ & $10^{-2}$ & $10^{-3}$ \\
        Standard deviation & $10^{-6}$ & $10^{-2}$ & $5\times10^{-3}$ & $5\times10^{-4}$ & $5\times10^{-4}$ & $5\times10^{-5}$ \\
        \hline
    \end{tabular}}
     \caption{Statistical properties of the disturbances}
\label{tab:disturbance_parameters}
\end{table}

Figures \ref{fig_r_ETCvsMPC_NL} and \ref{fig_v_ETCvsMPC_NL} show again the evolution of relative positions and velocities, respectively, using both the ETC and the MPC methods, where the simulations with MPC adopt the same the same propagation and disturbance models utilized for ETC and a prediction horizon equal to $N_p=N_{sim}/2$.


\begin{figure}[H]
    \centering
    \begin{minipage}[b]{0.48\textwidth}
        \centering
        \includegraphics[width=\textwidth]{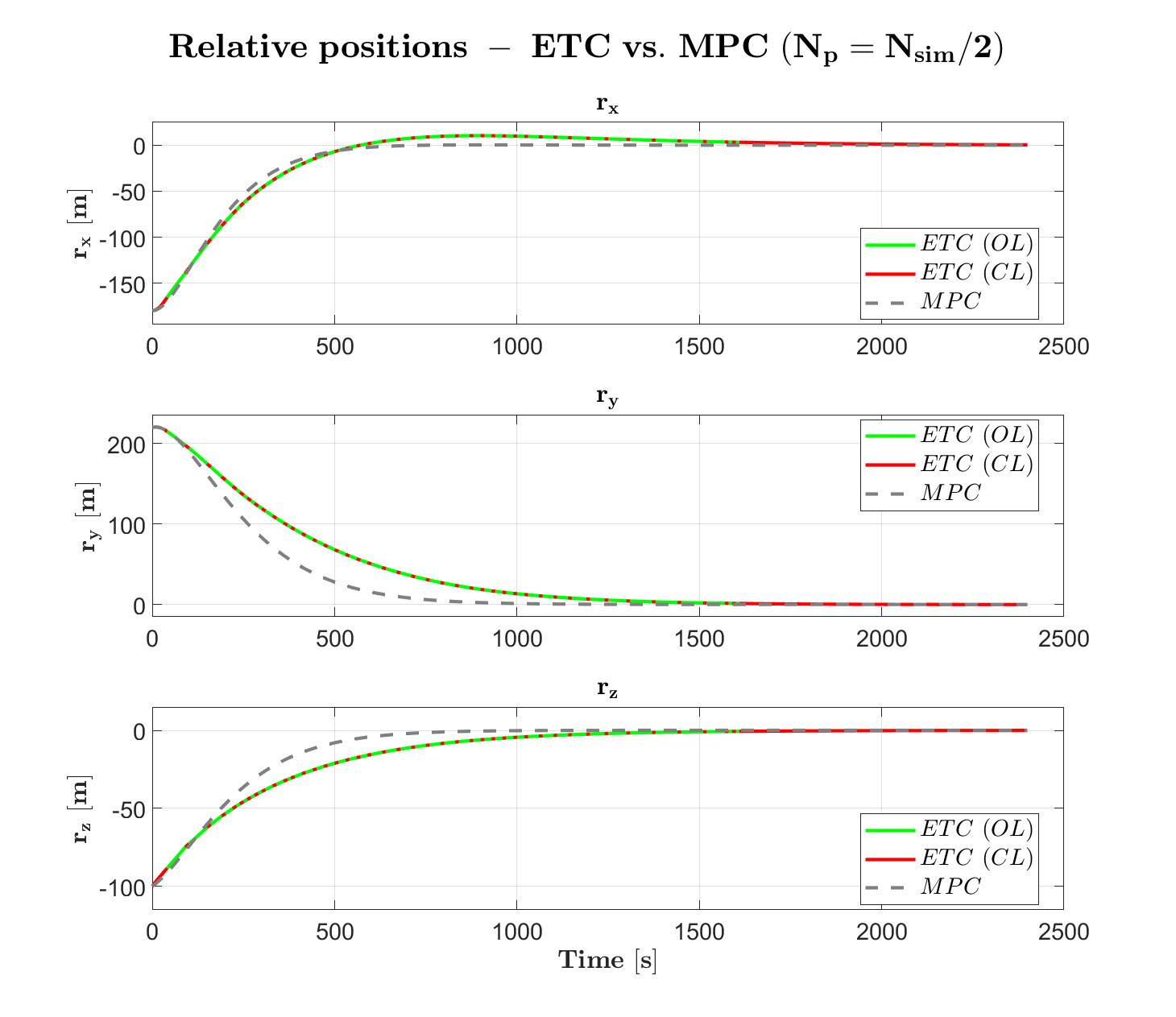}
        \caption{Evolutions of the chaser relative positions using the ETC (red: intervals following the application of an impulse; green: intervals without an initial impulse application) and MPC (gray dashed lines) strategies - Case 1 (single random realization) - Non-linear simulation model with imperfect thrust application and sensor noise}
        \label{fig_r_ETCvsMPC_NL}
    \end{minipage}
    \hfill
    \begin{minipage}[b]{0.48\textwidth}
        \centering
        \includegraphics[width=\textwidth]{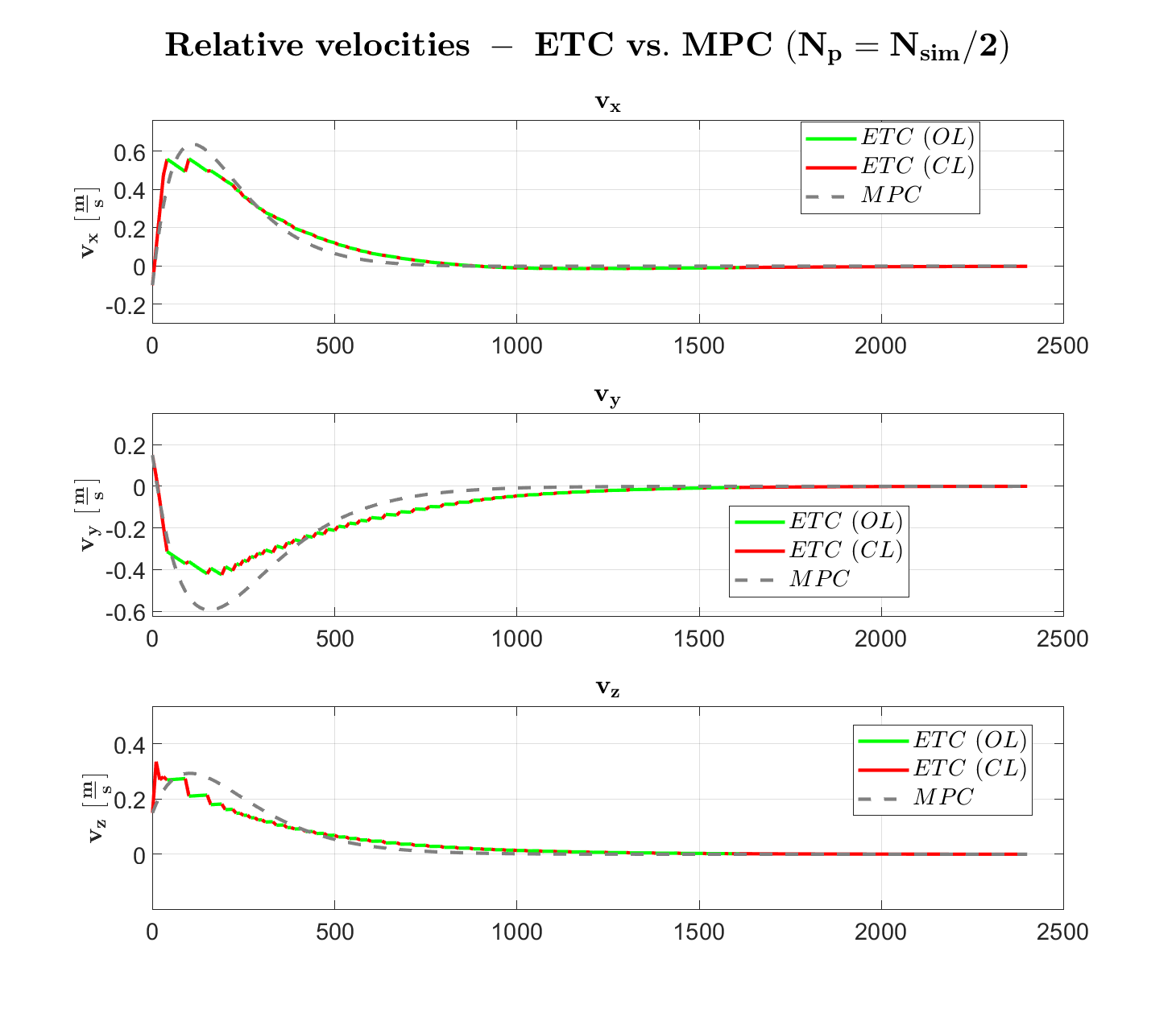}
        \caption{Evolutions of the chaser relative velocities using the ETC (red: intervals following the application of an impulse; green: intervals without an initial impulse application) and MPC (gray dashed lines) strategies - Case 1 (single random realization) - Non-linear simulation model with imperfect thrust application and sensor noise}
        \label{fig_v_ETCvsMPC_NL}
    \end{minipage}
\end{figure}

As in the linear case, the control strategies successfully drive the chaser to the target location. Table~\ref{tab:performance_indexes_NL} provides performance comparisons for this case:

\begin{table}[H]
    \centering
    \begin{tabular}{c|cc|cc|} 
        \cline{2-5}
         & \multicolumn{2}{c|}{$\bar{u}_{tot}\, [\mathrm{m/s}]$} 
         & \multicolumn{2}{c|}{$\% \text{ of firing instants}$} \\
        \cline{2-5}
         & ETC & MPC\ & ETC & MPC\\
        \hline
        \multicolumn{1}{|c|}{Case 1} & 2.623 & 3.089 &  51.67 & 100 \\
        \multicolumn{1}{|c|}{Case 2} & 2.224 & 2.330 &  70.83 & 100 \\
        \multicolumn{1}{|c|}{Case 3} & 1.957 & 2.451 &  48.33 & 100 \\
        \hline
    \end{tabular}
        \caption{Performance indexes: ETC vs MPC}
        \label{tab:performance_indexes_NL}
\end{table}

As can be seen from these results, the considerations regarding performances that were made for the linear case with simple additive disturbances extend also to this more complex scenario. Furthermore, the results demonstrate that the proposed ETC method exhibits strong robustness even in the presence of unmodeled dynamics and perturbations. In certain cases, the total energy expenditure and firing frequency are comparable to those observed in the linear case with additive disturbances, and in some instances even lower. This behavior can be attributed to the fact that, in the additive noise scenario, disturbances act continuously throughout the propagation intervals, whereas in the present case the chosen disturbances are applied only at the sampling instants. 

\section{Conclusions}\label{sec:conclusions}
This work presents an event-triggered control strategy for the spacecraft rendezvous problem, formulated within the HCW framework with impulsive inputs and explicitly accounting for input saturation and process noise. By introducing a switching mechanism based on state-dependent triggering conditions, the proposed method enables control inputs to be applied only when necessary, significantly reducing the number of firing instants compared to conventional periodic control approaches such as MPC. As a main feature, the control design is based on the solution of a single offline LMI optimization problem, through Lyapunov-based analysis, thereby extensively reducing the overall computational effort of the control implementation.\\
\indent Numerical simulations demonstrate the effectiveness of the proposed strategy across multiple rendezvous scenarios, taking into account different sources of disturbances. The results confirm that the event-triggered controller successfully guides the chaser to the target while achieving notable reductions in actuation frequency and comparable fuel efficiency relative to MPC.\\
\indent Future developments will focus on extending the proposed methodology to  a more realistic physical setting and to more satisfactory control solutions.
The inclusion of safety constraints, including line-of-sight \citep{doi:10.2514/1.29590} and collision avoidance, within the switching logic will be the main subject of future research. Furthermore, having considered in this work the simple scenario of a circular orbit, the adopted control strategy could be extended to the more complex case of an elliptical orbit (\citep{tschanuer},  \citep{Yamanaka}). Finally, future studies will investigate the inclusion of the Minimum Impulse Bit (MIB) constraint and its impact on feasibility and control performance.\\








\section*{Acknowledgments}
The work of T. del Carro and G. Portilla and A. Seuret are supported by grant ATR2023-145067 funded by MICIU/ AEI /10.13039/ 501100011033/. R. Vazquez acknowledges support of grant PID2023-147623OB-I00 funded by MICIU/AEI/10.13039/501100011033 and by ``ERDF A way of making Europe".

\section*{Declaration of Use of Artificial Intelligence}
Artificial intelligence was used exclusively to refine the English translation of this paper.

\bibliography{bibl_2026.bib}

\end{document}